\documentclass{article}

\usepackage{epsfig,xspace}
\usepackage{latexsym}
\usepackage[matrix,arrow,curve]{xy}\xyoption{all}

\usepackage[usenames,dvipsnames,svgnames,table]{xcolor}
\usepackage{amssymb,amsmath,latexsym,amsfonts,amsthm,alltt}
\usepackage[usenames,dvipsnames,svgnames,table]{xcolor}
\usepackage[dvips,usenames]{stmaryrd}
\usepackage{enumitem}
\usepackage{url}
\usepackage{graphicx}
\usepackage{float}
\usepackage{amstext}
\usepackage{cite}
\usepackage{hyperref}
\usepackage{MnSymbol}
\usepackage[left=2cm,right=2.5cm,top=2cm,bottom=1.5cm]{geometry}

\newtheorem{defn}{Definition}[section]
\newtheorem{rem}[defn]{Remark}
\newtheorem{thm}[defn]{Theorem}
\newtheorem{lemma}[defn]{Lemma}
\newtheorem{prop}[defn]{Proposition}
\newtheorem{coro}[defn]{Corollary}

\newtheorem{ex}{Example}[section]

\newcommand\crule[3][black]{\textcolor{#1}{\rule{#2}{#3}}}

\newcommand{\ra}{\rightarrow}
\newcommand{\lra}{\longrightarrow}

\newcommand{\Ra}{\Rightarrow}

\newcommand{\Da}{\Downarrow}

\newcommand{\midsp}{\;|\;}
\newcommand{\sub}[2]{#1_{{}_{#2}}}
\newcommand{\telos}{\hfill$\Box$}
\newcommand{\type}[1]{{\tt #1}}

\newcommand{\iso}{\backsimeq}

\newcommand{\val}[1]{\mbox{$[\![#1]\!]$}}
\newcommand{\forces}{\Vdash}
\newcommand{\dforces}{\forces^{\!\!\partial}}
\newcommand{\yvval}[1]{\mbox{$(\!|#1|\!) $}}

\newcommand{\infrule}[2]{\frac{\mbox{\rm $#1$}}{\mbox{\rm $#2$}}}
\newcommand{\proves}{\vdash}
\newcommand{\vproves}{\mbox{$\medvert\!\!\!\!\sim\;$}}
\newcommand{\vmodels}{\mbox{$\medvert\!\!\!\!\approx\;$}}
\newcommand{\zmodels}{\mbox{$\medvert\!\!\!\!\eqsim\;$}}

\newcommand{\upv}{\upVdash}
\newcommand{\rperp}{\mbox{${}^{\upv}$}}
\newcommand{\gphi}{{\mathcal  G}(Z_\partial)}
\newcommand{\gpsi}{{\mathcal  G}(Z_1)}

\newcommand{\bbox}{\blacksquare}

\newcommand{\lperp}{{}\rperp}

\newcommand{\ldd}{\mbox{$\largediamond\hspace*{-10pt}\Diamond\;$}}

\newcommand{\lbvert}{\mbox{\Large \mbox{$\boxvert$}}}
\newcommand{\lbminus}{\mbox{\Large \mbox{$\boxminus$}}}
\newcommand{\lbplus}{\mbox{\Large \mbox{$\boxplus$}}}

\newcommand{\ldvert}{\raisebox{0.5pt}{\Large \mbox{$\diamondvert$}}}

\newcommand{\ldminus}{\raisebox{0.5pt}{\Large \mbox{$\diamondminus$}}}

\newcommand{\blackdiamond}{\raisebox{-1.5pt}{\mbox{\LARGE {$\filleddiamond$}}}}
\newcommand{\lbb}{\mbox{\raisebox{1pt}{$\largesquare\hspace*{-7.6pt}\Box\;$}}}

\newcommand{\lbbox}{\raisebox{-1.2pt}[0pt][0pt]{\crule[black]{0.27cm}{0.27cm}}\hspace*{1pt}}
\newcommand{\lbboxi}{\sub{\lbbox}{I}}
\newcommand{\ldi}{\sub{\largediamond}{I}}

\newcommand{\stx}[2]{\mbox{ST$_{#1}(#2)$}}
\newcommand{\sty}[2]{\mbox{ST$_{#1}(#2)$}}

\newcommand{\lfspoon}{\leftfilledspoon}
\newcommand{\rfspoon}{\rightfilledspoon}

\newcommand{\rspoon}{\rightspoon}

\newcommand{\tright}{{\triangleright}}

\newcommand{\tdown}{{\triangledown}}
\newcommand{\btdown}{\boxtimes}
\newcommand{\lbtdown}{\mbox{\Large $\btdown$}}
\newcommand{\ttdown}{\mbox{$\tdown\hskip-7pt\smalltriangledown$}}

\newcommand{\ltdown}{\raisebox{1.1pt}{\mbox{$\largetriangledown$}}}

\newcommand{\lttdown}{\mbox{\raisebox{1.6pt}{$\largetriangledown\hskip-3.53mm\tdown$}}}

\newcommand{\dd}{{\diamonddiamond}}
\newcommand{\bb}{{\boxbox}}
\newcommand{\ccirc}{\circledcirc}
\newcommand{\ltright}{\largetriangleright}

\newcommand{\lset}{\langle}
\newcommand{\rset}{\rangle}

\title{Distribution-Free Modal Logics:\\ Sahlqvist -- Van Benthem Correspondence}

\author{Takis Hartonas \\
 University of Thessaly, Greece; hartonas@uth.gr}

\begin{document}
\maketitle

\begin{abstract}
We present an extension and generalization of Sahlqvist--Van Benthem correspondence to the case of distribution-free modal logic, with, or without negation and/or implication connectives. We follow a reductionist strategy, reducing the correspondence problem at hand to the same problem, but for a suitable system of sorted modal logic (the modal companion of the distribution-free system). The reduction, via a fully abstract translation, builds on  duality between normal lattice expansions and sorted residuated frames with relations (a generalization of classical Kripke frames with relations). The approach is scalable and it can be generalized to other systems, with or without distribution, such as distributive modal logic, or substructural logics with, or without additional modal operators.
\end{abstract}

\section{Introduction}
\label{intro}
There appears to be a gap, a discontinuity, in the correspondence theory for non-classical, distribution-free (modal, or otherwise) logics, in particular. A new paradigm in correspondence theory emerged, founded by Conradie, Ghilardi and Palmigiano  \cite{Conradie-unified} (2014), the {\em unified}, or {\em algebraic correspondence theory}. And with considerable success \cite{alg-dist,palmigiano-zhao,conradie-constructive-canonicity-fixed-point,conradie-craig-canonicity,Gehrke-Venema} over the last decade.
To a certain extent at least, it was influenced by the work of Conradie, Goranko and Vakarelov \cite{Conradie-Goranko-Vakarelov-I,Conradie-Goranko-Vakarelov-II,Conradie-Goranko-Vakarelov-III, Conradie-Goranko-IV,Conradie-Goranko-Vakarelov-V} and by the Ghilardi and Meloni insights in their work on constructive canonicity \cite{GHILARDI1997}. Its semantic underpinnings rest, essentially, on Gehrke's \cite{mai-gen} RS-frames semantics for distribution-free logics, continued and furthered in joint work with Palmigiano and coworkers \cite{mai-grishin,palmigiano-categories,palmigiano-to-meaning,palmigiano-cats2}.
 Published literature leaves the reader with the impression that the classical Sahlqvist -- Van Benthem \cite{van-b,SAHLQVIST1975} approach to correspondence, eliminating second-order quantifiers by a minimal instantiation argument, has little to contribute in a new, non-classical, even distribution-free setting. Noticeable exceptions are Suzuki's \cite{SUZUKI_2013} and a much clearer presentation in a recent article by Bezhanishvili, Dmitrieva, de Groot and Morachini \cite{choice-free-dmitrieva-bezanishvili} on distribution-free positive modal logic, but the signature of the logic is limited and it remains to be demonstrated how the approach can be generalized and extended to richer systems.
 
 The present article contributes by restoring (while generalizing) the classical Sahlqvist -- Van Benthem  approach to correspondence theory in a distribution-free setting and it reduces to it if the frames we work with, and their associated logics, are classical. We demonstrate the results by working with distribution-free modal logic, equipped  with both negation and  implication logical operators and we briefly discuss possible extensions (to substructural logics, for example). We adopt a uniform relational semantics approach based on representation and (topological) duality for normal lattice expansions and sorted residuated frames with relations \cite{duality2}, a generalization of the J\'{o}nsson and Tarski \cite{jt1} representation of Boolean algebras with operators, with particular cases of \cite{duality2} treated in \cite{pnsds,choiceFreeHA,choiceFreeStLog,dfnmlA}. The proposed semantics reduces to classical semantics in the special case where our frames, and their associated logics, are classical.
 
 The present article is part of a project of reduction of non-distributive logics to sorted residuated modal logics \cite{pll7,redm,pnsds,discr,discres,choiceFreeHA,choiceFreeStLog,dfnmlA,vb} (their sorted modal companion logics), relying heavily on duality \cite{sdl,dloa,sdl-exp,duality2} for lattice expansions with normal operators. 
 
 Our correspondence argument is mediated by a translation and a dual-translation (co-translation) of the language of the target logic in the (sorted) language of an extension of its sorted companion modal logic, in which adjoint and dual operators are added. Whereas the target logic (distribution-free modal logic, in the present article) is the logic of the full complex algebras of frames in a broad frame-class, the sorted companion modal logic is the logic of the dual sorted powerset algebras of the frames.
 
 Pursuing Sahlqvist correspondence via translation is not an idea that is new with this article. In their article {\em Sahlqvist via Translation}, Conradie, Palmigiano and Zhao \cite{palmigiano-zhao} carried out this approach for distributive systems with some known G\"{o}del type translation into their modal companions. As the authors remind their reader, the idea had been already around, in one form or another \cite{Gehrke-Venema,vbenthem-bezh-hol,wolter-zakhar,wolter-zakha-in-book}.
 \\

Section~\ref{logics etc section} presents the language and proof system of Distribution-free Modal Logic ($\mathrm{DfML}$, the logic of main focus in this article), as well as its equivalent modal lattice algebraic semantics and it briefly describes the basics of sorted residuated frames, interpretations and models. 

Section~\ref{lang section} presents a suitably chosen language for the dual sorted powerset algebra of a frame, together with the interpretation of the language in sorted residuated frames. The first and second-order frame languages are also briefly described and the last part of the section is devoted to the definition of a translation and a co-translatioin (dual translation) of the language of DfML into the language of its sorted companion modal logic. A full abstraction result for the translation is reported, as an instance of a more general result proven in \cite[Theorem~1]{vb}. The section is completed with a brief description of the standard first and second-order fully abstract translation of the language of the companion sorted modal logic.

The extension of the Sahlqvist -- Van Benthem correspondence theory to distribution-free modal logic (with, or without additional negation and/or implication connectives) is presented in Section~\ref{sahlqvist section}. 

Section~\ref{reduction section} presents the reduction pre-processing stage, reducing the translation, or co-translation of a given sequent to what we define to be a system of formal inequalities in canonical Sahlqvist form. 

Section~\ref{reduction structure section} describes the main steps of the algorithm, from pre-processing to generation of the guarded second-order translation and, finally, to elimination of second-order quantifiers and termination with a local first-order correspondent.

To simplify the presentation we treat separately the case of DfML with necessity and possibility connectives only (Section~\ref{box and diamond section}), with a negation connective added (Section~\ref{negation section}), or with implication as well (Section~\ref{implication section}). Most of the needed work is presented in the first part, Section~\ref{box and diamond section}, with only box and diamond in the language. The reduction strategy is presented and a detailed proof of the extended Sahlqvist -- Van Benthem correspondence result is given in Theorem~\ref{Sahlqvist thm}, subsequently extended for negation and implication as well (Theorem~\ref{Sahlqvist thm with negation}, Theorem~\ref{Sahlqvist thm with implication}). 

In the last Section~\ref{etc} we first substantiate the claim that the generalized Sahlqvist -- Van Benthem approach of this article reduces to the classical result (Section~\ref{kripke case}) if our frames and their associated logics are classical, we give enough pointers for the interested reader to work out an extension of this article's approach to substructural logics (Section~\ref{suzuki case}), we briefly comment on the Conradie-Palmigiano \cite{conradie-palmigiano} approach, pointing out that not every sequent for which a local first-order correspondent can be calculated in this article's approach is Sahlqvist in the approach taken in \cite{conradie-palmigiano}.

We leave out any discussion of canonicity in this article, postponing it for a future report. There appear to exist some inherent difficulties in establishing canonicity in a ``via translation'' approach and this has been pointed out already in \cite{palmigiano-zhao}.

\section{Logics, Algebras, Frames and Models}
\label{logics etc section}
\subsection{Implicative Modal Lattices and Logics}
\label{lat and log section}
Let $\{1,\partial\}$ be a 2-element set, $\mathbf{L}^1=\mathbf{L}$ and $\mathbf{L}^\partial=\mathbf{L}^\mathrm{op}$ (the opposite lattice). Extending established terminology \cite{jt1}, a function $f:\mathbf{L}_1\times\cdots\times\mathbf{L}_n\lra\mathbf{L}_{n+1}$ will be called {\em additive} and {\em normal}, or a {\em normal operator}, if it distributes over finite joins of the lattice $\mathbf{L}_i$, for each $i=1,\ldots n$, delivering a join in $\mathbf{L}_{n+1}$.

An $n$-ary operation $f$ on a bounded lattice $\mathbf L$ is {\em a normal lattice operator of distribution type  $\delta(f)=(j_1,\ldots,j_n;j_{n+1})\in\{1,\partial\}^{n+1}$}  if it is a normal additive function  $f:{\mathbf L}^{j_1}\times\cdots\times{\mathbf L}^{j_n}\lra{\mathbf L}^{j_{n+1}}$ (distributing over finite joins in each argument place), where  each $j_k$, for  $k=1,\ldots,n+1$,   is in the set $\{1,\partial\}$, hence ${\mathbf L}^{j_k}$ is either $\mathbf L$, or ${\mathbf L}^\partial$.

If $\tau$ is a tuple (sequence) of distribution types, a {\em normal lattice expansion (NLE) of (similarity) type $\tau$} is a lattice with a normal lattice operator of distribution type $\delta$ for each $\delta$ in $\tau$. 

In this article we consider normal lattice expansions $\mathbf{L}=(L,\leq,\wedge,\vee,0,1,\Box,\Diamond,\tdown,\ra)$ where $\delta(\Box)=(\partial;\partial)$, $\delta(\Diamond)=(1;1)$, $\delta(\tdown)=(1;\partial)$ and $\delta(\ra)=(1,\partial;\partial)$. In other words, the following axioms are assumed to hold, on top of the axioms for bounded lattices,
\begin{tabbing}
\hskip5mm\=(D$\Box$)\hskip5mm\= $\Box(a\wedge b)=\Box a\wedge\Box b$\hskip1cm\=(D$\Diamond$)\hskip5mm\=$\Diamond(a\vee b)=\Diamond a\vee\Diamond b$\\
\>(N$\Box$) \> $\Box 1=1$\>(N$\Diamond$) \> $\Diamond 0=0$\\
\>(N$\tdown$)\> $\tdown 0=1$ \> (D$\tdown$)\> $\tdown(a\vee b)=\tdown a\wedge\tdown b$
\end{tabbing}
\begin{tabbing}
\hskip5mm\=(A1)\hskip5mm\= $(a\vee b)\ra c=(a\ra c)\wedge(b\ra c)$\\
\>(A2)\> $a\ra(b\wedge c)=(a\ra b)\wedge(a\ra c)$\\
\>(N)\> $(0\ra a) = 1 = (a\ra 1)$
\end{tabbing}
We refer to algebras as above as {\em implicative modal lattices with a quasi-complementation operation}, or briefly as {\em modal lattices}.

The propositional language of modal lattices is defined by the grammar below
\[
\mathcal{L}_{\Box\Diamond}^{\tdown\ra}\ni\varphi\;:=\;p_i(i\in\mathbb{N})\midsp\top\midsp\bot\midsp\varphi\wedge\varphi\midsp \varphi\vee\varphi\midsp\bb\varphi\midsp\dd\varphi\midsp\ttdown\varphi\midsp\varphi\rfspoon\varphi.
\]
A proof system for distribution-free modal logic (DfML), the logic of modal lattices, is defined in Table~\ref{proof system}, in terms of single-premiss single-conclusion sequents, written as $\varphi\proves\psi$. It is left to the interested reader to verify that the Lindenbaum-Tarski algebra of DfML is a modal lattice. 

\begin{table}[!htbp]
\caption{A Proof System for Minimal Distribution-Free Modal Logic}
\label{proof system}
\begin{tabbing}
$\varphi\proves\varphi$  \hskip2.5cm\= $\varphi\proves\varphi\vee\psi$ \hskip3cm\= $\psi\proves\varphi\vee\psi$\\
$\varphi\proves\top$\> $\varphi\wedge\psi\proves\varphi$ \> $\varphi\wedge\psi\proves\psi$\\
$\bot\proves\varphi$ \> $\infrule{\varphi\proves\psi}{\varphi[\vartheta/p]\proves\psi[\vartheta/p]}$ \> $\infrule{\varphi\proves\psi\hskip4mm\psi\proves\vartheta}{\varphi\proves\vartheta}$\\[1mm]
\> $\infrule{\varphi\proves\vartheta\hskip4mm\psi\proves\vartheta}{\varphi\vee\psi\proves\vartheta}$ \> $\infrule{\varphi\proves\psi\hskip4mm\varphi\proves\vartheta}{\varphi\proves\psi\wedge\vartheta}$
\\[2mm]
$\dd\bot\proves\bot$ \> $\dd(\varphi\vee\psi)\proves\dd\varphi\vee\dd\psi$ \> $\infrule{\varphi\proves\psi}{\dd\varphi\proves\dd\psi}$\\
$\top\proves\bb\top$ \> $\bb\varphi\wedge\bb\psi\proves\bb(\varphi\wedge\psi)$ \> $\infrule{\varphi\proves\psi}{\bb\varphi\proves\bb\psi}$
\\
$\top\proves\ttdown\bot$ \> $\ttdown\varphi\wedge\ttdown\psi\proves\ttdown(\varphi\vee\psi)$ \> $\infrule{\varphi\proves\psi}{\ttdown\psi\proves\ttdown\varphi}$
\\
$\top\proves\bot\rfspoon\varphi$ \> $\top\proves\varphi\rfspoon\top$
\\[2mm]
$\varphi\vee\psi\rfspoon\vartheta\proves(\varphi\rfspoon\vartheta)\wedge(\psi\rfspoon\vartheta)$
\hskip6mm
$(\vartheta\rfspoon\varphi)\wedge(\vartheta\rfspoon\psi)\proves\vartheta\rfspoon\varphi\wedge\psi$
\\[2mm]
$\infrule{\proves\varphi}{\overline{\top\proves\varphi}}$ \> $\infrule{\varphi\proves\psi}{\proves\varphi\rfspoon\psi}$
\>
$\infrule{\psi\proves\varphi\hskip4mm\vartheta\proves\chi}{\varphi\rfspoon\vartheta\proves\psi\rfspoon\chi}$
\end{tabbing}
\end{table}
For a detailed discussion on the algebraic semantics we refer the reader to \cite{dfnmlA}. Our interest in the current article is more with relational semantics in sorted frames, briefly reviewed from \cite{duality2,dfnmlA} in the next section.

\subsection{Sorted Residuated Frames and Models}
\label{frames section}
\subsubsection{Relational Structures}
\begin{defn}\label{frame defn}
By a (relational) frame we mean a structure $\mathfrak{F}=(s,Z,I,(R_j)_{j\in J},\sigma)$, where $s=\{1,\partial\}$ is a set of sorts, $Z=(Z_t)_{t\in s}$ is a nonempty sorted set ($Z_t\neq\emptyset$, for each $t\in s$), where we make no assumption of disjointness of sorts, $I\subseteq\prod_{t\in s}Z_t$ is a distinguished sorted relation, $\sigma$ is a sorting map on $J$ with  $\sigma(j)\in s^{n(j)+1}$ and $(R_j)_{j\in J}$ is a family of sorted relations such that if $\sigma(j)=(j_{n(j)+1};j_1,\ldots,j_{n(j)})$, then $R_j\subseteq Z_{j_{n(j)+1}}\times\prod_{k=1}^{n(j)}Z_{j_k}$. 

The {\em sort $\sigma(R_j)$ (or just $\sigma(j)$, or $\sigma_j$) of the relation} $R_j$ is the tuple $\sigma(j)=(j_{n(j)+1};j_1\cdots j_{n(j)})$. 
\end{defn}

We often display the sort of a relation as a superscript, as in $\mathfrak{F}=(s,Z,I,(R^{\sigma(j)}_j)_{j\in J})$. For example, $R^{11}, T^{\partial 1\partial}$ designate sorted relations $R\subseteq Z_1\times Z_1$ and $T\subseteq Z_\partial\times Z_1\times Z_\partial$. 

In the intended application of the present article the frame relations considered are $R^{11}_\Diamond, R^{\partial\partial}_\Box$, $R^{\partial 1}_\tdown$ and $R^{\partial 1\partial}_\ra$, but we use $T$ for the latter, or $T^{\partial 1\partial}$ (displaying its sort), in order to make it easier to relate to results obtained in \cite{choiceFreeHA,dePaiva-Bierman2000}. 

Sorted frames collapse to classical Kripke frames when $Z_1=Z_\partial$ and $I$ is the identity relation (consult \cite[Remark~3.4, Remark~3.9]{dfnmlA} for details).
  
\subsubsection{The Underlying Polarity of a Frame -- The Lattices of Stable and Co-stable Sets}
The relation $I$ generates a residuated pair $\lambda:\powerset(Z_1)\leftrightarrows\powerset(Z_\partial):\rho$, defined as usual by
\[
\lambda(U)=\{y\in Z_\partial\midsp\exists x\in Z_1(xIy\wedge x\in U)\}\hskip1cm
\rho(V)=\{x\in Z_1\midsp\forall y\in Z_\partial(xIy\lra y\in V)\}.
\]
We may also use the notation $\sub{\largediamond}{I} U$ for $\lambda U$ and $\sub{\lbbox}{I} V$ for $\rho V$, as we have often done in previous published work. 

The complement of $I$ will be designated by $\upv$ and we refer to it as the {\em Galois relation of the frame}. It generates a Galois connection $(\;)\rperp:\powerset(Z_1)\leftrightarrows\powerset(Z_\partial)^{\rm op}:\rperp(\;)$ defined by
\[
U\rperp=\{y\in Z_\partial\midsp\forall u\in Z_1 (u\in U\lra u\upv y)\}\hskip1cm
\rperp V=\{x\in Z_1\midsp\forall y\in Z_\partial(y\in V\lra x\upv y)\}.
\]

Observe that the closure operators generated by the residuated pair and the Galois connection are identical, i.e. $\rho\lambda U=\rperp(U\rperp)$ and $\lambda\rho V=(\rperp V)\rperp$. This follows from the fact that $U\rperp=\sub{\largesquare}{I}(-U)$ and ${}\rperp V=\sub{\lbbox}{I}(-V)$. 

To simplify, we often use a priming notation for both Galois maps $(\;)\rperp$ and $\rperp(\;)$, i.e. we let $U'=U\rperp$, for $U\subseteq Z_1$, and $V'=\rperp V$, for $V\subseteq Z_\partial$. Hence $U''=\rperp(U\rperp)=\rho\lambda U$ and $V''=(\rperp V)\rperp=\lambda\rho V$. 

\begin{defn}\label{Galois set lattice}
The complete lattice of all {\em Galois stable} sets $Z_1\supseteq U=U''$ will be designated by $\mathcal{G}(Z_1)$ and the complete lattice of all {\em Galois co-stable} sets $Z_\partial\supseteq V=V''$ will be similarly denoted by $\mathcal{G}(Z_\partial)$. We refer to Galois stable and co-stable sets as {\em Galois sets}.
\end{defn} 

Note that each of $Z_1, Z_\partial$ is a Galois set, but the empty set need not be Galois. Clearly, the Galois connection is a dual isomorphism $(\;)':\gpsi\iso\gphi^\partial:(\;)'$. 

For an element $u$ in either $Z_1$ or $Z_\partial$ and a subset $W$, respectively of $Z_\partial$ or $Z_1$, we write $u|W$, under a well-sorting assumption, to stand for either $u\upv W$ (which stands for $u\upv w$, for all $w\in W$), or $W\upv u$ (which stands for $w\upv u$, for all $w\in W$), where well-sorting means that either $u\in Z_1, W\subseteq Z_\partial$, or $W\subseteq Z_1$ and $u\in Z_\partial$, respectively. Similarly for the notation $u|v$, where $u,v$ are elements of different sort. We occasionally decorate the subset relation with a sort subscript, writing for example $B\subseteq_\partial D$ when $B,D\in\gphi$ and $B\subseteq D$.

A preorder relation is defined on each of $Z_1,Z_\partial$ by $u\preceq w$ iff $\{u\}'\subseteq\{w\}'$. We call a frame {\em separated} if $\preceq$ is in fact a partial order $\leq$. For an element $u$ (of either $Z_1$ or $Z_\partial$) we write $\Gamma u$ for the set of elements $\preceq$-above it. We hereafter assume that frames are separated.

Sets $\Gamma w$ and $\{w\}'$ will be referred to as {\em principal elements}.  $\Gamma w$ will be referred to as a {\em closed element} and $\{w\}'$ as an {\em open element}. 

The following basic facts will be often used without reference to  Lemma~\ref{basic facts}.
\begin{lemma}
\label{basic facts}
Let $\mathfrak{F}=(s,Z,I,(R_j)_{j\in J},\sigma)$ be a  frame, $u\in Z=Z_1\cup Z_\partial$ and $\upv$ the Galois relation of the frame. Let $v|G$ refer to either $G\upv v$, if $G\in\gpsi, v\in Z_\partial$, or $v\upv G$, if $v\in G_1$ and $G\in\gphi$.
\begin{enumerate}
\item $\upv$ is increasing in each argument place (and thereby its complement $I$ is decreasing in each argument place).
\item $(\Gamma u)'=\{u\}'$ and $\Gamma u=\{u\}^{\prime\prime}$ is a Galois set.
\item Galois sets are increasing, i.e. $u\in G$ implies $\Gamma u\subseteq G$.
\item For a Galois set $G$, $G=\bigcup_{u\in G}\Gamma u$.
\item For a Galois set $G$, $G=\bigvee_{u\in G}\Gamma u=\bigcap_{v|G}\{v\}'$.
\item For a Galois set $G$ and any set $W$, $W^{\prime\prime}\subseteq G$ iff $W\subseteq G$.
\end{enumerate}
\end{lemma}
\begin{proof}
  By simple calculation. Proof details are included in \cite[Lemma 2.2]{sdl-exp}.  For claim 4, $\bigcup_{u\in G}\Gamma u\subseteq G$ by claim 3 (Galois sets are upsets). 
\end{proof}

We let $\vec{u}[\;]_k$ be the vector with a hole (or just a place-holder) at the $k$-th position and write $u[w]_k$ either to display the element at the $k$-th place, or to designate the result of filling the $k$-th place of $u[\;]_k$, or to denote the result of replacing the element $u_k$ in $\vec{u}$ by the element $w$.

For $1\leq k\leq n$, the {\em $k$-th section of an $(n+1)$-ary relation $S$} is the set $wS\vec{u}[\;]_k$. For $k=n+1$ the section is simply the set $S\vec{u}=\{w\midsp wS\vec{u}\}$. 

\subsubsection{Galois Dual and Double Dual Relations -- Smoothness and Frame Axioms}
\begin{defn}
\label{Galois dual relation}
For a sorted $(n+1)$-ary frame relation $R_j$, its {\em Galois dual relation} $R^\prime_j$ is defined by $R_j^\prime u_1\cdots u_n=(R_ju_1\cdots u_n)'$.
\end{defn} 
For the relations $R^{11}_\Diamond, R^{\partial\partial}_\Box, R^{\partial 1}_\tdown$ and $T^{\partial 1\partial}$ we will also have use of their double-duals, defined below.

\begin{defn}\label{double dual relations}
The {\em double dual relations} $R''_\Diamond, R''_\Box, R''_\tdown$ and $R^{111}$ of $R^{11}_\Diamond, R^{\partial\partial}, R^{\partial 1}_\tdown$ and $T^{\partial 1\partial}$, respectively, are defined as follows.
\begin{enumerate}
\item[($R''_\Diamond$)] Set $yR_\Diamond''=\lperp(yR_\Diamond')=(yR'_\Diamond)'$ (and recall that the Galois dual $R_\Diamond'$ is defined from $R_\Diamond\subseteq Z_1\times Z_1$ by setting $R_\Diamond'z=(R^{11}_\Diamond z)\rperp=(R^{11}_\Diamond z)'$). 
\item[($R''_\Box$)] Similarly, for the {\em double dual} $R_\Box''$ defined from $R_\Box\subseteq Z_\partial\times Z_\partial$ by first letting $R_\Box'y=\lperp(R_\Box y)$ be the Galois dual relation, then setting $xR_\Box''=\lperp(xR_\Box')$.
\item[($R''_\tdown$)] Similarly for the double dual $R''_\tdown$ defined by first letting $R'_\tdown x=(R^{\partial 1}x)'$, then setting for any $z\in Z_1$, $zR''_\tdown=(zR'_\tdown)'\subseteq Z_1\times Z_\partial$.
\item[($R^{111}$)]\ \ Define  $R^{111}\subseteq Z_1\times(Z_1\times Z_1)$ from the frame relation $T^{\partial 1\partial}$ by first setting $T'zv=(T^{\partial 1\partial}zv)'$, then permuting arguments to define $vR^{\partial 11}zx$ iff $xT^{11\partial}zv$ and lastly, taking the Galois dual relation of $R^{\partial 11}$, by setting $R^{111}zx=(R^{\partial 1 1}zx)'$.
\end{enumerate}
\end{defn}

\begin{rem}[Double duals in Kripke Frames]
\label{double-duals in Kripke}
  If the frame $\mathfrak{F}=(s,Z,I,(F_j)_{j\in J})$ is a Kripke frame, i.e. $Z_1=Z_\partial$ and $I\subseteq Z_1\times Z_\partial$ is the identity relation (hence $x\upv y$ iff $x\neq y$) it was pointed out in \cite[Remark~3.4,Remark~3.9]{dfnmlA} that the Galois connection is set-complementation, hence every subset is Galois. Then also the Galois dual relation is a complement relation, e.g. $R'_\Diamond z=-R^{11}_\Diamond z$. It then follows that double-dual relations are identical to the original relations. Indeed,  
  \begin{tabbing}
 $yR''_\Diamond v$ \hskip3mm\= iff\hskip2mm\= for any $z$, if $yR'_\Diamond z$, then $z\neq v$\\
 \>iff\> for any $z$, if $z=v$, then it is not the case that $y$ is in $-R^{11}_\Diamond z$\\
 \>iff\> $yR^{11}_\Diamond v$
  \end{tabbing}
  Similar arguments apply to the other double-dual relations, with $R^{111}$ and $T^{\partial 1\partial}$ only differing by the permutation of arguments involved in the definition of $R^{111}$.
  
  This has the further consequence that the semantics we define in Table~\ref{sat} collapses to classical semantics, with co-interpretation $y\dforces\varphi$ being simply $y\not\forces\varphi$ and boxes and diamonds being interpreted classically. See also Section~\ref{kripke case} for more on this issue and on its significance for our correspondence approach.
\end{rem}

\begin{defn}
\label{smooth defn}
Call a frame relation $R_j$ {\em smooth} iff every section of its Galois dual relation $R_j^\prime$ is a Galois set (stable, or co-stable, according to the sort $\sigma(R_j)$ of the relation). 
\end{defn}

Hereafter, when considering a structure $\mathfrak{F}=(s,Z,I,(R_j)_{j\in J},\sigma)$ we always assume that the frame is separated and that all frame relations are smooth.
Since no other kind of frame (relational structure) will be considered in this article, we shall refer to relational structures $\mathfrak{F}=(s,Z,I,(R_j)_{j\in J},\sigma)$ simply as frames, or sorted residuated frames.

Table~\ref{frames axioms}  lists a minimal set of frame axioms that we need, the separation and smoothness axioms.

\begin{table}[!htbp]
\caption{Axiomatization of Frames $\mathfrak{F}=(s,Z,I,(R_j)_{j\in J},\sigma)$}
\label{frames axioms}
($R_j\in\{R^{\partial\partial}_\Box,R^{11}_\Diamond,R^{\partial 1}_\tdown,T^{\partial 1\partial}\}$)
\begin{tabbing}
  (F1)\hskip0.5cm \=  The frame is separated. \\
  (F2)  \>   Every frame relation $R_j$ is smooth.
\end{tabbing}
\hrule
\end{table}

For convenience only we will also consider adding the following two axioms.
\begin{tabbing}
(F0)\hskip0.5cm \= 
  The frame relation $I$  is quasi-serial, i.e. the conditions\\
  \>  $\forall x\in Z_1\exists y\in Z_\partial\; xIy$ and $\forall y\in Z_\partial\exists x\in Z_1\; xIy$ hold.\\
(F3)\> Every frame relation is increasing in the left (first) argument place and decreasing \\
  \> in every other argument place.
\end{tabbing}
 Note that (F0) enforces that the empty set is a Galois set. All four axioms are canonical (they hold in the canonical frame of a modal lattice) and they are part of the axiomatization of frames in the topological duality argument of \cite{duality2}.

\subsubsection{Relational Semantics for DfML}
A relational model $\mathfrak{M}=(\mathfrak{F},V)$ consists of a frame $\mathfrak{F}$ and a sorted valuation $V=(V^1,V^\partial)$ of propositional variables, interpreting a variable $p$ as a Galois stable set $V^1(p)\in\gpsi$ and co-interpreting it as a Galois co-stable set $V^\partial(p)=V^1(p)\rperp\in\gphi$.  Interpretations and co-interpretations determine each other in the sense that for any sentence $\varphi\in\mathcal{L}_\tau$, if $\val{\varphi}\in\mathcal{G}(Z_1)$ is an interpretation extending a valuation $V^1$ of propositional variables as stable sets, then $\val{\varphi}\rperp=\yvval{\varphi}\in\mathcal{G}(Z_\partial)$  is the co-interpretation extending the valuation $V^\partial$.

Satisfaction ${\forces}\subseteq Z_1\times\mathcal{L}_\tau$ and co-satisfaction (refutation) ${\dforces}\subseteq Z_\partial\times\mathcal{L}_\tau$ relations are then defined as expected, by $Z_1\ni x\forces\varphi$ iff $x\in\val{\varphi}$ and $Z_\partial\ni y\dforces\varphi$ iff $y\in\yvval{\varphi}$. Satisfaction and co-satisfaction determine each other, in the sense that $x\forces\varphi$ iff $\forall y(y\dforces\varphi\lra x\upv y)$ and $y\dforces\varphi$ iff $\forall x(x\forces\varphi\lra x\upv y)$.
Therefore, for each operator it suffices to provide either its satisfaction, or its co-satisfaction (refutation) clause. We do this in Table \ref{sat}, in line with the principle of order-dual relational semantics introduced in \cite{odigpl}.

\begin{table}[!htbp]
\caption{(Co)Satisfaction relations}
\label{sat}
\begin{tabbing}
$x\forces p_i$\hskip8mm\=iff\hskip3mm\= $x\in V^1(p_i)$
\hskip3cm\= 
$y\dforces p_i$    \hskip8mm\=iff\hskip3mm\= $V^1(p_i)\upv y$     \\
$x\forces\top$ \>iff\> $x=x$ \>  $y\dforces\bot$\>iff\> $y=y$\\
$x\forces\varphi\wedge\psi$\>iff\> $x\forces\varphi$ and $x\forces\psi$ \> $y\dforces\varphi\vee\psi$ \>iff\> $y\dforces\varphi$ and $y\dforces\psi$\\
$x\forces \bb\varphi$ \>iff\> $\forall z\in Z_1\;(xR_\Box''z\lra z\forces\varphi)$
\>  $y\dforces \dd\varphi$ \hskip6mm\>iff\> $\forall v\in Z_\partial\;(yR_\Diamond''v\lra v\dforces\varphi)$
\\
$x\forces\ttdown\varphi$\>iff\> $\forall y\in Z_\partial(xR''_\tdown y\lra y\dforces\varphi)$ 
  \\
$x\forces\varphi\rfspoon\psi$ \>iff\> $\forall u,z\in Z_1(u\forces\varphi\wedge zR^{111}ux\lra z\forces\psi)$ 
\end{tabbing}
\hrule
\end{table}

\begin{rem}\label{alternative sat}
The  clause $x\forces\varphi\rfspoon\psi$ iff $\forall u\in Z_1\forall y\in Z_\partial(u\forces\varphi\;\wedge\; y\dforces\psi\lra xT'uy)$ was also presented in \cite[Table~2]{dfnmlA} as an alternative to modeling implication and a proof of equivalence with the clause presented in Table~\ref{sat} was given in \cite[Proposition~3.11]{dfnmlA}. For the weak negation operator $\ttdown$, the Galois dual $R'_\tdown=\bot$ of the frame relation $R^{\partial 1}_\tdown$ was used in \cite[Section~3.2]{choiceFreeStLog}, yielding the satisfaction clause $x\forces\ttdown\varphi$ iff $\forall z\in Z_1(z\forces\varphi\lra x\bot z)$, but the two clauses can be easily verified to be equivalent by the interested reader.
\end{rem}

\subsection{Sorted Powerset Algebras and Full Complex Algebras of Frames}
\label{dual algebras section}
While the language of distribution-free modal logic (with negation and implication) is the language of the full complex algebras of frames (which are implicative modal lattices as these were defined in Section~\ref{lat and log section}), 
the {\em reduction language} that we use in our generalized Sahlqvist -- Van Benthem result is the language of the dual sorted powerset algebras of frames, the structure of which is detailed in this section.

Given a frame $\mathfrak{F}=(s,Z,I,(R_j)_{j\in J},\sigma)$, each relation $R_j\subseteq Z_{j_{n(j)+1}}\times\prod_{k=1}^{n(j)}Z_{j_k}$ generates a sorted image operator, defined as in the Boolean case, except for the sorting
\begin{align}\label{sorted image ops}
F_j(\vec{W})&=\;\{w\in Z_{i_{n(j)+1}}\midsp \exists \vec{w}\;(wR_j\vec{w}\wedge\bigwedge_{s=1}^{n(j)}(w_s\in W_s))\} &=\; \bigcup_{\vec{w}\in\vec{W}}R_j\vec{w}.
\end{align}

\begin{defn}\label{sorted powerset algebra defn}
The {\em dual sorted powerset algebra} of a frame $\mathfrak{F}=(s,Z,I,(R_j)_{j\in J},\sigma)$ is the algebra $\mathbf{P}=((\;)':\powerset(Z_1)\leftrightarrows\powerset(Z_\partial):(\;)',(F_j)_{j\in J} )$, where for each $j\in J$, $F_j$ is the sorted image operator  generated by the frame relation $R_j$ by~\eqref{sorted image ops}. 
\end{defn}

Equation~\eqref{sorted image ops} specializes in our case of interest to \eqref{ldvert}--\eqref{ltright}, where $U\subseteq Z_1, V\subseteq Z_\partial$,
\begin{eqnarray}
\ldvert U &=& \{x\in Z_1\midsp\exists z\in Z_1(xR^{11}_\Diamond z\wedge z\in U)\}\label{ldvert}\\
\ldminus V &=& \{y\in Z_\partial\midsp \exists v\in Z_\partial(yR^{\partial\partial}_\Box v\wedge v\in V)\}\label{ldminus}\\
\ltdown U &=& \{y\in Z_\partial\midsp\exists x\in Z_1(yR^{\partial 1}_\tdown x\wedge x\in U)\}\label{ltdown}\\
U{\largetriangleright} V&=&\{y\in Y\midsp\exists x,v(x\in U\;\wedge\;v\in V\;\wedge\;yTxv)\}, \label{ltright}
\end{eqnarray}
resulting in the dual powerset algebra $\mathbf{P}=((\;)':\powerset(Z_1)\leftrightarrows\powerset(Z_\partial):(\;)', \ldminus,\ldvert,\ltdown,\ltright )$. 

The propositional language of the dual powerset algebra of a frame is displayed below
\begin{eqnarray*}
\mathcal{L}_1\ni\alpha,\eta,\zeta &=& P_i (i\in\mathbb{N})\;\midsp\top\midsp\bot\midsp \alpha\cap\alpha\midsp\alpha\cup\alpha\midsp{\diamondvert}\alpha\midsp \beta'
\\
\mathcal{L}_\partial\ni\beta,\delta,\xi &=& P^i(i\in\mathbb{N})\midsp\top\midsp\bot\midsp \beta\cap\beta\midsp\beta\cup\beta\midsp{\diamondminus}\beta
\midsp{\tdown}\alpha\midsp\alpha'\midsp \alpha\tright\beta 
\end{eqnarray*}
The language is interpreted in frames in a standard way, with satisfaction clauses based on equations \eqref{ldvert}--\eqref{ltright}. 

\begin{rem}
The language displayed is the language of the companion sorted modal logic of DfML, equipped with three unary and one binary (sorted) normal additive (diamond) operators. DfML can be embedded in a full and faithful way in its companion sorted modal logic (consult Theorem~\ref{full abstraction of trans in sorted modal}), with a syntactic translation that parallels the representation of operators in the frame. 

For our correspondence argument we will work with an extension of the sorted modal language, as we will in effect have use for residuals and/or Galois dual operators of the above image operators, as well. In addition, right residuals in the powerset algebra restrict to operations on Galois sets, a fact of which we make use in our correspondence argument.
\end{rem}

For a  subset $W$ of $Z_1$ or $Z_\partial$ let $\overline{W}=W''$ be its closure and
if $F_j$ is the (sorted) image operator generated by the frame relation $R_j$, let $\overline{F}_j$ be the closure of the restriction of $F_j$ to Galois sets (stable, or co-stable, according to sort).
\[
\xymatrix{
\prod_{k=1}^{n(j)}\powerset(Z_{j_k})\ar^{F_j}[rr]\ar^{(\;)''}@<0.5ex>@{->>}[d] && \powerset(Z_{j_{n(j)+1}})\ar^{(\;)''}@<0.5ex>@{->>}[d]\\
\prod_{k=1}^{n(j)}\mathcal{G}(Z_{j_k})\ar^{\overline{F}_j}[rr]\ar@<0.5ex>@{^{(}->}[u] && \mathcal{G}(Z_{j_{n(j)+1}})\ar@<0.5ex>@{^{(}->}[u]
}
\]

\begin{thm}
\label{dist from section stability}
The sorted operator $\overline{F}_j:\prod_{k=1}^{n(j)}\mathcal{G}(Z_{j_k})\lra \mathcal{G}(Z_{j_{n(j)+1}})$ distributes over arbitrary joins of Galois sets, in each argument place, returning a join in $\mathcal{G}(Z_{j_{n(j)+1}})$.
\end{thm}
\begin{proof}
The claim was proven in \cite[Theorem~3.12]{duality2}, using the smoothness property of the frame relation $R_j$.
\end{proof}

Let $\mathbf{G}=((\;)':\gpsi\iso\gphi^\partial:(\;)', (\overline{F}_j)_{j\in J})$ be the sorted algebra of Galois sets and observe that by Theorem~\ref{dist from section stability} $\overline{(\;)}:\mathbf{P}\lra\mathbf{G}$ is a homomorphism of sorted algebras, indeed an epimorphism, taking intersections to intersections, unions to joins (closures of unions) and normal additive operators $F_j$ to normal additive operators $\overline{F}_j$. 

By the complete distribution property, $\overline{F}_j$ is residuated at each argument place and, from residuation, it follows that $\overline{F}_j$ is normal, i.e. $\overline{F}_j(\vec{G}[\emptyset]_k)=\emptyset$.

Note that for each $j\in J$ the sorted set operator $F_j:\prod_{k=1}^{n(j)}\powerset(Z_{j_k})\lra\powerset(Z_{j_{n(j)+1}})$ in the (sorted) powerset algebra $\mathbf{P}$ is completely additive (it distributes over arbitrary unions) in each argument place. Hence it is residuated, i.e. for each $1\leq k\leq n(j)$ there exists a set map $G_{j,k}$ such that $F_j(\vec{W}[V]_k)\subseteq U$ iff $V\subseteq G_{j,k}(\vec{W}[U]_k)$ which is defined by equation~\eqref{k-residual}
\begin{equation}\label{k-residual} 
G_{j,k}(\vec{W}[U]_k)=\bigcup\{V\subseteq Z_{j_k}\midsp F_j(\vec{W}[V]_k)\subseteq U\}.
\end{equation}

We let $\lbbox_1,\lbbox_\partial$ be the right residuals of $\ldvert,\ldminus$, respectively. We may occasionally drop the subscripts $1,\partial$, letting context determine which residual is meant.

\begin{thm}\label{preservation of residuals}
If $G_{j,k}$ is the right residual of $F_j$ at the $k$-th argument place, then its restriction to Galois sets is the right residual $\overline{G}_{j,k}$ of $\overline{F}_j$ at the $k$-th argument place. Letting $P,Q,E$ range over Galois sets (and $\vec{P},\vec{Q}$ over tuples thereof) the right $k$-residual $\overline{G}_{j,k}$  of $\overline{F}_j$ can be defined in any of the equivalent ways in equation \eqref{residuals def}
  \begin{equation}\label{residuals def}
  \begin{array}{ccl}
  \overline{G}_{j,k}(\overline{P}[Q]_k) &=& \bigcup\{E\in\mathcal{G}(Z_{j_k})\midsp F_j(\vec{P}[E]_k)\subseteq Q\} \\
  &=&  \bigcup\{\Gamma u\in\mathcal{G}(Z_{j_k})\midsp F_j(\vec{P}[\Gamma u]_k)\subseteq Q\}  \\
  &=& \{u\in Z_{j_k}\midsp F_j(\vec{P}[\Gamma u]_k)\subseteq Q\}.
  \end{array}
  \end{equation}
\end{thm}
\begin{proof}
Consult \cite[Theorem~3.7]{dfnmlA}.
\end{proof}

\begin{coro}
The (sorted) homomorphism $\overline{(\;)}:\mathbf{P}\lra\mathbf{G}$ preserves any residuation facts that obtain in $\mathbf{P}$. Moreover, if $F_j\dashv G_{j,k}$, then the restriction of $G_{j,k}$ to Galois sets returns a Galois set (i.e. the restriction is identical to the closure $\overline{G}_{j,k}$ of the restriction ) and $\overline{F}_j\dashv\overline{G}_{j,k}$.
\end{coro}

The Galois connection is a dual isomorphism of the complete lattices of stable and co-stable sets, $(\;)':\mathcal{G}(Z_1)\iso\mathcal{G}(Z_\partial)^\mathrm{op}:(\;)'$. This allows for extracting single-sorted operators $\overline{F}_j^1:\prod_{k=1}^{n(j)}\mathcal{G}(Z_1)\lra\mathcal{G}(Z_1)$ and $\overline{F}_j^\partial:\prod_{k=1}^{n(j)}\mathcal{G}(Z_\partial)\lra\mathcal{G}(Z_\partial)$, by composition with the Galois connection maps
\begin{equation}\label{2single-sorted}
  \overline{F}^1_j(A_1,\ldots,A_{n(j)})=
  \left\{
  \begin{array}{cl}
  \overline{F}_j(\ldots,\underbrace{A_k}_{j_k=1},\ldots,\underbrace{A'_r}_{j_r=\partial},\ldots) & \mbox{if } j_{n(j)+1}=1
  \\
  (\overline{F}_j(\ldots,\underbrace{A_k}_{j_k=1},\ldots,\underbrace{A'_r}_{j_r=\partial},\ldots))' & \mbox{if } j_{n(j)+1}=\partial.
  \end{array}
  \right.
\end{equation}
The dual operators $\overline{F}_j^\partial:\prod_{k=1}^{n(j)}\mathcal{G}(Z_\partial)\lra\mathcal{G}(Z_\partial)$ are similarly defined by
\begin{equation}\label{2single-sorted dual}
  \overline{F}^\partial_j(B_1,\ldots,B_{n(j)})=
  \left\{
  \begin{array}{cl}
  \overline{F}_j(\ldots,\underbrace{B_k}_{j_k=\partial},\ldots,\underbrace{B'_r}_{j_r=1},\ldots) & \mbox{if } j_{n(j)+1}=\partial
  \\
  (\overline{F}_j(\ldots,\underbrace{B_k}_{j_k=\partial},\ldots,\underbrace{B'_r}_{j_r=1},\ldots))' & \mbox{if } j_{n(j)+1}=1.
  \end{array}
  \right.
\end{equation}
Illustrating for a unary map $F$, we obtain $\overline{F}^1(A)=(\overline{F}^\partial(A'))'$ and conversely. Since $\overline{F}^\partial(B)=(F(B))''$, it follows that $\overline{F}^1(A)=(\overline{F}^\partial(A'))'=(F(A')'$.

In our case of interest, $\overline{F}^1_j$ are the maps defined on $A,C\in\gpsi$ by 
\begin{itemize}
\item
$\lbb A=(\ldminus A')'$ and ${\ldd}A=(\ldvert A)''$, for the modal operators,  
\item $\lttdown A=(\ltdown A)'$, for quasi-complementation, and 
\item $A\Ra C=(A\ltright C')'$, for implication. 
\end{itemize}
Normality and the complete (co)distribution properties of the operators $\lbb,\ldd,\lttdown$ and $\Ra$ are a consequence of Theorem~\ref{dist from section stability}.

\begin{defn}\label{full complex algebra defn}
For a frame $\mathfrak{F}=(s,Z,I,(R_j)_{j\in J},\sigma)$, its full complex algebra $L(\mathfrak{F})$ is defined as the normal lattice expansion $L(\mathfrak{F})=\mathfrak{F}^+=(\gpsi,\subseteq,\bigcap,\bigvee,\emptyset, Z_1,(\overline{F}^1_j)_{j\in J})$. In particular for the case of current interest, the full complex algebra is the algebra $\mathfrak{F}^+=(\gpsi,\subseteq,\bigcap,\bigvee,\emptyset, Z_1,\lbb,\ldd,\lttdown,\Ra)$.
\end{defn}

By Theorem~\ref{preservation of residuals}, right adjoints restricted to Galois sets return Galois sets. We make repeated use of this fact in the sequel, for implication, necessity and quasi-complement. 

\paragraph{Implication.}
Let $\bigodot:\powerset(Z_1)\times\powerset(Z_1)\lra \powerset(Z_1)$ be the binary image operator generated by the double dual relation $R^{111}$ of $T^{\partial 1\partial}$ (Definition~\ref{double dual relations}), defined on $U,V\subseteq Z_1$ by
\begin{equation}\label{odot defn}
U\mbox{$\bigodot$}V=\{x\in Z_1\midsp\exists u,z(u\in U\wedge z\in V\wedge xR^{111}uz)\}=\bigcup_{u\in U}^{z\in V}R^{111}uz.
\end{equation}
Being completely additive, $\bigodot$ is residuated (in both argument places)
and we let $\Ra_T$ be the residual $U\Ra_T W=\bigcup\{X\midsp U\bigodot X\subseteq W\}$, satisfying the residuation condition $V\subseteq U\Ra_T W$ iff $U\bigodot V\subseteq W$. Since $x\in U\Ra_T W$ iff $U\bigodot\{x\}\subseteq W$, a straightforward calculation, left to the interested reader, gives 
\begin{equation}\label{Lra membership condition}
  U\Ra_T W=\{x\in Z_1\midsp\forall z,u\in Z_1(u\in U\wedge zR^{111}ux\lra z\in W)\}
\end{equation}

Let $\bigovert$ be the closure of the restriction of $\bigodot$ to stable sets. 

\begin{prop}
\label{Ra long and short}
  The operators $\bigovert,\Ra$ are residuated in $\gpsi$. In other words, for any Galois stable sets $A,F,C$ we have $A\bigovert F\subseteq C$ iff $F\subseteq A\Ra C$. Consequently, $\Ra$ is the restriction to Galois stable sets of the residual $\Ra_T$ of $\bigodot$.
\end{prop}
\begin{proof}
 Consult \cite[Proposition~3.11]{dfnmlA}.
\end{proof}
Proposition~\ref{Ra long and short} justifies our henceforth using the notation $\Ra$ for either the powerset operator, or the Galois set operator.

\paragraph{Necessity.}
The box operator $\lbb$ in the full complex algebra of a frame was defined on a stable set $A\in\gpsi$ as the Galois dual operator of (the closure of) the diamond operator $\ldminus$, by $\lbb A=(\ldminus A')'$. 

\begin{prop}\label{box as restriction}
The stable set operator $\lbb$ is the restriction to stable sets of the powerset dual image operator $\lbminus:\powerset(Z_1)\lra\powerset(Z_1)$ generated by the relation $R''_\Box\subseteq Z_1\times Z_1$. 
\end{prop}
\begin{proof}
Consult \cite[Proposition~3.13]{dfnmlA}.
\end{proof}
\begin{coro}\label{boxvert as restriction}
The Galois dual $(\ldvert B')'$ of the diamond operator $\ldvert:\powerset(Z_1)\lra\powerset(Z_1)$ can be similarly defined and the analogue of Proposition~\ref{box as restriction} holds for the dual image operator $\lbvert:\powerset(Z_\partial)\lra\powerset(Z_\partial)$ generated by the relation $R''_\Diamond\subseteq Z_\partial\times Z_\partial$, obtaining  $(\ldvert B')'$ as the restriction of $\lbvert$ to Galois co-stable sets, i.e. the identity $\lbvert B= (\ldvert B')'$ holds. \telos
\end{coro}

\paragraph{Quasi-Complement.}
The relation $R''_\tdown\subseteq Z_1\times Z_\partial$ (double dual of the relation $R^{\partial 1}_\tdown$) generates a sorted box operator $\lbtdown:\powerset(Z_\partial)\lra\powerset(Z_1)$, defined as usual by $\lbtdown V=\{x\in Z_1\midsp xR''_\tdown\subseteq V\}$. 

\begin{prop}\label{ltdown and black}
  The restriction of $\lbtdown$ to $\gphi$ (the lattice of co-stable sets) is the Galois dual operation of $\ltdown$, i.e. the following identities hold
  \[
(\ltdown A)'=\lttdown A =\lbtdown A'\;\mbox{ and } \lttdown B'=\lbtdown B=(\ltdown B')'.
\]
\end{prop}
\begin{proof}
The following calculation
\begin{tabbing}
$x\in\lttdown A$\hskip6mm\= iff\hskip2mm\= $\forall y\in Z_\partial(y\in\ltdown A\lra x\upv y)$\\
\>iff\> $\forall y\in Z_\partial[\exists z\in Z_1(yR^{\partial 1}_\tdown z\wedge z\in A)\lra x\upv y]$\\
\>iff\> $\forall z[z\in A\lra R^{\partial 1}_\tdown z\subseteq\{x\}\rperp]$\\
\>iff\>  $\forall z[z\in A\lra\Gamma x\subseteq R'_\tdown z]$\\
\>iff\> $\forall z[z\in A\lra xR'_\tdown z]$ iff $ A\subseteq xR'_\tdown$  \\
\>iff\> $xR''_\tdown\subseteq A'$  \\
\>iff\> $x\in\lbtdown A'$
\end{tabbing} 
completes the proof.
\end{proof}

Table~\ref{interaction identities} reviews the identities established. 

\begin{table}[!htbp]
\caption{Interaction Identities}
\label{interaction identities}
($A,F,C\in\gpsi, B\in\gphi, U,V\subseteq Z_1$)
\begin{tabbing}
$(\ltdown A)'=\lttdown A =\lbtdown A'$ \hskip1cm\= $\lttdown B'=\lbtdown B=(\ltdown B')'$\\[2mm]
$\lbminus A=\lbb A=(\ldminus A')'$ \> $\lbvert B=(\ldvert B')'$\hskip2cm\= $\ldd A=(\ldvert A)''$\\[2mm]
$A\Ra C =(A\ltright C')'=A\Ra_TC$ \>\> $A\bigovert C=(A\bigodot C)''$
\\[2mm]
$V\subseteq U\Ra_T(U\bigodot V)$ \> $U\bigodot(U\Ra_T V)\subseteq V$
\\[2mm]
$F\subseteq A\Ra(A\bigovert F)$ \> $A\bigovert(A\Ra F)\subseteq F$
\\[2mm]
$\ldvert\lbbox_1 A\subseteq A\subseteq \lbbox_1\ldvert A$ \> $\ldminus\lbbox_\partial B\subseteq B\subseteq \lbbox_\partial\ldminus B$
\end{tabbing}
\hrule
\end{table}

Note that $\lbb:\gpsi\lra\gpsi$ is defined as the Galois dual of $\ldminus:\powerset(Z_\partial)\lra\powerset(Z_\partial)$. If the frame is a classical Kripke frame ($Z_1=Z_\partial$ and with $I$ the identity relation) then the Galois connection maps $(\;)'$ are the set-complement map and the definition of $\lbb$ from $\ldminus$ collapses to the classical definition (consult \cite[Remark~3.4 and Remark~3.9]{dfnmlA} for details).

\begin{defn}\label{extended powerset algebra defn}
The {\em extended sorted powerset algebra} of $\mathfrak{F}=(s,Z,I,R^{11}_\Diamond,R^{\partial\partial}_\Box,R^{1\partial}_\tdown, T^{\partial 1\partial})$ is the structure
\[\mathbf{P}=((\;)':\powerset(Z_1)\leftrightarrows\powerset(Z_\partial):(\;)',\mbox{$\bigodot$},\Ra, \lbminus,\lbvert,\ldminus,\lbbox_\partial,\ldvert,\lbbox_1,\ltdown,\lbtdown,\ltright,\Da).\]
\end{defn}

\begin{thm}[Soundness of DfML]\label{soundness for DfML}
  DfML is sound in frames $\mathfrak{F}=(s,Z,I,R^{11}_\Diamond,R^{\partial\partial}_\Box,R^{1\partial}_\tdown, T^{\partial 1\partial})$ axiomatized in Table~\ref{frames axioms}.
\end{thm}
\begin{proof}
  By a standard argument, given the results of the present section. Details are left to the interested reader.
\end{proof}

\section{Languages and Translations}
\label{lang section}
\subsection{Sorted Residuated Modal Logic}
\label{modal companions section}
The sorted modal logic of the extended dual powerset (modal) algebra of a frame $\mathfrak{F}$ is (an extension of the minimal) sorted modal companion of distribution-free implicative modal logic with weak negation. We display its full syntax below.  As indicated in the previoius section, a subset of the language is only needed in order to fully and faithfully embed the language of DfML, but we shall need the full language displayed below as a reduction language for our correspondence argument.
\begin{eqnarray*}
\mathcal{L}_1\ni\alpha,\eta,\zeta &=& P_i (i\in\mathbb{N})\;\midsp\top\midsp\bot\midsp \alpha\cap\alpha\midsp\alpha\cup\alpha\midsp{\diamondvert}\alpha\midsp\bbox_1\alpha\midsp\boxminus\alpha 
\midsp{\btdown}\beta\midsp\beta'\midsp \alpha{\odot}\alpha\midsp\alpha\rspoon\alpha\\
\mathcal{L}_\partial\ni\beta,\delta,\xi &=& P^i(i\in\mathbb{N})\midsp\top\midsp\bot\midsp \beta\cap\beta\midsp\beta\cup\beta\midsp{\diamondminus}_\partial\beta\midsp\bbox_\partial\beta\midsp\boxvert\beta 
\midsp{\tdown}\alpha\midsp\alpha'\midsp \alpha\tright\beta 
\end{eqnarray*}

Given a structure $\mathfrak{F}$, an interpretation of $\mathcal{L}_s$ is a sorted function $V=(V_1,V_\partial)$ assigning a subset $V_1(P_i)\subseteq Z_1$ and $V_\partial(P^i)\subseteq Z_\partial$ to the propositional variables of each sort.  A model $\mathfrak{M}$ is a pair $\mathfrak{M}=(\mathfrak{F},V)$ consisting of a frame and an interpretation of sorted propositional variables in the structure, as above.

Given a model $\mathfrak{M}$, its interpretation function generates a sorted satisfaction relation $\zmodels=(\models,\vmodels)$, where ${\models}\subseteq Z_1\times\mathcal{L}_1$ and ${\vmodels}\subseteq Z_\partial\times\mathcal{L}_\partial$, defined in Table~\ref{sorted sat table}. 

\begin{table}[!htbp]
\caption{Sorted satisfaction relation, given a model $\mathfrak{M}=(\mathfrak{F},V)$}
\label{sorted sat table}
 ($u\in Z_1, v\in Z_\partial$)
\begin{tabbing}
$u\models P_i$\hskip7mm\=iff\hskip2mm\= $u\in V_1(P_i)$\hskip3.2cm\= $v\vmodels P^i$ \hskip6mm\=iff\hskip2mm\= $v\in V_\partial(P^i)$
\\[2mm]
$u\models\top$ \>iff\> $u=u$ \> $v\vmodels\top$\>iff\> $v=v$
\\[2mm]
$u\models\bot$ \>iff\> $u\neq u$ \> $v\vmodels\bot$ \>iff\> $v\neq v$
\\[2mm]
$u\models\alpha\cap\eta$\>iff\> $u\models\alpha$  and $u\models\eta$   \>  $v\vmodels\beta\cap\delta$ \>iff\> $v\vmodels\beta$ and $v\vmodels\delta$
\\[2mm]
$u\models\alpha\cup\eta$ \>iff\> $u\models\alpha$ or $u\models\eta$
    \>
    $v\vmodels\beta\cup\delta$ \>iff\> $v\vmodels\beta$ or $v\vmodels\delta$
\\[2mm]
$u\models\beta'$ \>iff\> $\forall y\in Z_\partial(uIy\lra y\not\vmodels\beta)$
    \>
    $v\vmodels\alpha'$ \>iff\> $\forall u\in Z_1(uIy\lra u\not\models\alpha)$
\\[2mm]
$u\models\diamondvert\alpha$ \>iff\>    $\exists z\in Z_1(uR^{11}_\Diamond z\wedge z\models\alpha)$
    \>
    $v\vmodels{\diamondminus}\beta$ \>iff\>   $\exists y\in Z_\partial(vR^{\partial\partial}_\Box y\wedge y\vmodels\beta)$
\\[2mm]
$u\models\bbox_1\alpha$ \>iff\> $\forall z\in Z_1(zR^{11}_\Diamond u\lra z\models\alpha)$    \> $v\vmodels\bbox_\partial\beta$ \>iff\> $\forall y\in Z_\partial(yR^{\partial\partial}_\Box v\lra y\vmodels\beta)$
\\[2mm]
$u\models\boxminus\alpha$ \>iff\> $\forall z\in Z_1(uR''_\Box z\lra z\models\alpha)$
\>
$v\vmodels\boxvert\beta$ \>iff\> $\forall y\in Z_\partial(vR''_\Diamond y\lra y\vmodels\beta)$    
\\[2mm]
$u\models\btdown\beta$ \>iff\>  $\forall y\in Z_\partial(uR''_\tdown y\lra y\vmodels\beta)$    
\> $v\vmodels\tdown\alpha$ \>iff\>  $\exists x\in Z_1(yR^{\partial 1}_\tdown x\wedge x\models\alpha)$      
\\[2mm]
$u\models\alpha\odot\eta$ \>iff\> $\exists x,z\in Z_1(uR^{111}xz\wedge x\models\alpha\wedge z\models\eta)$   \> 
\\[2mm]
$u\models\alpha\rspoon\eta$ \>iff\> $\forall x,z\in Z_1(x\models\alpha\wedge zR^{111}xu  \lra z\models\eta)$
\\[2mm]
\>\>\hskip3cm $v\vmodels\alpha\tright\beta\;$ iff $\;\exists x\in Z_1\exists y\in Z_\partial(vT^{\partial 1\partial}xy\wedge x\models\alpha\wedge y\vmodels\beta)$
\end{tabbing}
\hrule
\end{table}

For $\alpha\in\mathcal{L}_1$ and $\beta\in\mathcal{L}_\partial$, we let $\val{\alpha}_\mathfrak{M}=\{x\in Z_1\midsp x\models\alpha\}$ and $\yvval{\beta}_\mathfrak{M}=\{y\in Z_\partial\midsp y\vmodels\beta\}$. If the model $\mathfrak{M}$ is understood from context, we omit the subscript.

Proof-theoretic consequence may be defined as a sorted relation ${\Ra} =(\proves,\vproves)$, with ${\proves}\subseteq\mathcal{L}_1\times\mathcal{L}_1$ and ${\vproves}\subseteq\mathcal{L}_\partial\times\mathcal{L}_\partial$.  We refer to sequents $\alpha\proves\eta$ as {\em 1-sequents} and to sequents $\beta\vproves\delta$ as {\em $\partial$-sequents}. Validity of a sequent in a model is defined as usual. We shall have no use of a proof system and, given a class $\mathbb{C}$ of frames, we may define the logic $\Lambda(\mathbb{C})$ as the sorted set of 1-sequents and $\partial$-sequents validated in every frame $\mathfrak{F}\in\mathbb{C}$. For our purposes, we let $\mathbb{C}$ be the class of sorted frames axiomatized as in Table~\ref{frames axioms}.

\subsection{Sorted First and Second-Order Language of Frames}
The first-order language of frames $\mathfrak{F}=(s,Z,I,(R_j)_{j\in J},\sigma)$ is the sorted language 
\[
\mathcal{L}^1_s(V_1,V_\partial,(\mathbf{P}_i)_{i\in\mathbb{N}},(\mathbf{P}^i)_{i\in\mathbb{N}},\mathbf{I},(\mathbf{R}^{\sigma(j)}_j)_{j\in J}, \wedge_1,\wedge_\partial,=_1,=_\partial)
\]
with 
\begin{itemize}
\item sorted variables $v^1_0,v^1_1,\ldots\in V_1$, $v^\partial_0,v^\partial_1,\ldots\in V_\partial$ 
\item unary predicates $\mathbf{P}_i, \mathbf{P}^i$, for each sort, respectively, and for $i\in\mathbb{N}$
\item a binary predicate $\mathbf{I}$ of sort $(1,\partial)$
\item a sorted predicate $\mathbf{R}_j$, for each $j\in J$ (a countable set), of sort $\sigma(j)$ and arity $n(j)+1$
\item identity predicates $=_1,=_\partial$ for each sort.
\end{itemize}
Terms (sorted) are first-order variables and meet terms obtained by applying the operations $\wedge_1,\wedge_\partial$ to terms of the appropriagte sort.  Atomic formulae of each sort are defined as usual from the predicates and terms, respecting sorting. Formulae are built recursively, closing the set of atomic formulae under negation, conjunction and sorted universal quantification $\forall^1 v^1_i\Phi(v^1_i), \forall^\partial v^\partial_i\Phi(v^\partial_i)$, observing sorting. By well formed (defined as in the single-sorted case) we mean, in addition, well-sorted. 

Given a sorted valuation $V$ of individual variables,  $\mathfrak{F}\models_s\Phi[V]$ is defined as in the case of unsorted FOL. When $V(u^1_k)=a\in Z_1$, we may also display the assignment in writing $\mathfrak{F}\models_s\Phi(u^1_k)[u^1_k:=a]$, or just $\mathfrak{F}\models_s\Phi(u^1_k)[a]$, and similarly for more variables. 

The second-order language $\mathcal{L}^2_s$ of frames $\mathfrak{F}=(s,Z,I,(R_j)_{j\in J},\sigma)$ is the sorted language 
\[
\mathcal{L}^2_s(V_1,V_\partial,({P}_i)_{i\in\mathbb{N}},({P}^i)_{i\in\mathbb{N}},\mathbf{I},(\mathbf{R}^{\sigma(j)}_j)_{j\in J},\wedge_1,\wedge_\partial,=_1,=_\partial)
\]
with countable sets of second-order variables $P_i, P^i$, with $i\in\mathbb{N}$, one for each sort, interpreted as subsets of $Z_1,Z_\partial$, according to sort. Terms and well-formed formulae are defined accordingly, in the obvious way. We use the same symbols $\forall^1,\forall^\partial$ for second-order quantification as for first-order.

In the current context of interest, the languages contain the binary predicates $\mathbf{R}^{11}_\Diamond,\mathbf{R}^{\partial\partial}_\Box$, $\mathbf{R}^{\partial 1}_\tdown$ and $\mathbf{T}^{\partial 1\partial}$ in addition to (the binary predicate) $\mathbf{I}$ and the sorted identity predicates $=_1,=_\partial$. We will typically drop the sorting superscripts, as understood.

Defined (in the expected way) predicates include $\widetilde{\mathbf{I}}$, to be interpreted as the Galois relation $\upv$ of the frame, $\mathbf{R}'_\Box, \mathbf{R}'_\Diamond,\mathbf{R}'_\tdown,\mathbf{T}'$, to be interpreted as the Galois dual relations of the relations interpreting the predicates $\mathbf{R}^{11}_\Diamond,\mathbf{R}^{\partial\partial}_\Box$, $\mathbf{R}^{\partial 1}_\tdown$ and $\mathbf{T}^{\partial 1\partial}$, as well as predicates $\mathbf{R}''_\Diamond, \mathbf{R}''_\Box, \mathbf{R}''_\tdown$ and $\mathbf{R}^{111}$, to be interpreted as the double dual relations of the relations interpreting the predicates $\mathbf{R}^{11}_\Diamond,\mathbf{R}^{\partial\partial}_\Box$, $\mathbf{R}^{\partial 1}_\tdown$ and $\mathbf{T}^{\partial 1\partial}$.

\subsection{Language Translations}

Table~\ref{syntactic translation into sorted} defines by mutual recursion a syntactic translation $(\;)^\bullet$ and co-translation $(\;)^\circ$ of the language $\mathcal{L}$ of modal lattices into the language $\mathcal{L}_s=(\mathcal{L}_1,\mathcal{L}_\partial)$ of sorted modal logic.

\begin{table}[!htbp]
\caption{Definition of the syntactic translation and co-translation of the language of modal lattices}
\label{syntactic translation into sorted}
\begin{tabbing}
\hspace*{1cm}\=$p_i^\bullet$ \hskip1.5cm\==\hskip4mm\= $P_i''$ \hskip3cm\= $p_i^\circ$\hskip1.5cm\==\hskip4mm\= $P_i'$\\
\> $\top^\bullet$\>=\>$\top$ \> $\top^\circ$\>=\> $\bot$\\
\> $\bot^\bullet$ \>=\> $\bot$ \> $\bot^\circ$ \>=\> $\top$\\
\> $(\varphi\wedge\psi)^\bullet$ \>=\> $\varphi^\bullet\cap\psi^\bullet$ \> $(\varphi\wedge\psi)^\circ$ \>=\> $(\varphi^\circ\cup\psi^\circ)''$\\
\> $(\varphi\vee\psi)^\bullet$ \>=\> $(\varphi^\bullet\cup\psi^\bullet)''$ \> $(\varphi\vee\psi)^\circ$ \>=\> $\varphi^\circ\cap\psi^\circ$
\\
\> $(\bb\varphi)^\bullet$ \>=\> $(\diamondminus\varphi^\circ)'$ \> $(\bb\varphi)^\circ$ \>=\> $(\diamondminus\varphi^\circ)''$
\\
\>\>=\> $\boxminus\varphi^\bullet$   
\\
\> $(\dd\varphi)^\bullet$ \>=\> $(\diamondvert\varphi^\bullet)''$ \> $(\dd\varphi)^\circ$ \>=\> $(\diamondvert\varphi^\bullet)'$
\\
\>\>\>\>\>=\> $\boxvert\varphi^\circ$
\\
\> $(\ttdown\varphi)^\bullet$ \>=\> $(\tdown\varphi^\bullet)'$ \> $(\ttdown\varphi)^\circ$ \>=\> $(\tdown\varphi^\bullet)''$
\\ 
\> \>=\> $\btdown\varphi^\circ$
\\
\> $(\varphi\rfspoon\psi)^\bullet$ \>=\> $(\varphi^\bullet\tright\psi^\circ)'$ \> $(\varphi\rfspoon\psi)^\circ$ \>=\> $(\varphi^\bullet\tright\psi^\circ)''$ 
\\
\>\>=\> $\varphi^\bullet\rspoon\psi^\bullet$
\\[2mm]
\> Translation and Co-translation of Sequents
\\
\> $(\varphi\proves\psi)^\bullet$ \>=\> $\varphi^\bullet\proves\psi^\bullet$
\>
$(\varphi\proves\psi)^\circ$ \>=\> $\psi^\circ\vproves\varphi^\circ$
\end{tabbing}
\end{table}

Both languages $\mathcal{L}$ and $\mathcal{L}_s$ are interpreted in frames $\mathfrak{F}=(s,Z,I,R^{11}_\Diamond, R^{\partial\partial}_\Box,R^{1\partial}_\tdown,T^{\partial 1\partial})$. Given a model $\mathfrak{M}=(\mathfrak{F},V)$ for $\mathcal{L}_s$, a model $\mathcal{N}=(\mathfrak{F},\bar{V})$ for $\mathcal{L}$ is obtained by setting $\bar{V}^1(p_i)=V_1(P_i)''$, generating an interpretation and a co-interpretation of $\mathcal{L}$-sentences.

A sentence $\alpha\in\mathcal{L}_1$ is a {\em classical modal correspondent} of a sentence $\varphi\in\mathcal{L}$ iff for any $\mathcal{L}_s$-model $\mathfrak{M}=(\mathfrak{F},V)$, $\val{\alpha}_\mathfrak{M}=\val{\varphi}_\mathfrak{N}$, where $\mathfrak{N}$ is defined as above.

\begin{thm}[Full Abstraction]
\label{full abstraction of trans in sorted modal}
Let $\mathfrak{M}=(\mathfrak{F},V)$ be a model of the sorted modal language $\mathcal{L}_s=(\mathcal{L}_1,\mathcal{L}_\partial)$. Then, for any sentence $\varphi\in\mathcal{L}$,
\begin{enumerate}
\item  its translation $\varphi^\bullet$ is a classical modal correspondent of $\varphi$. In other words, $\val{\varphi^\bullet}_\mathfrak{M}=\val{\varphi}_\mathfrak{N}= \val{(\varphi^\circ)'}_\mathfrak{M}=\val{(\varphi^\bullet)''}_\mathfrak{M}$
\item  $\yvval{\varphi^\circ}_\mathfrak{M}=\yvval{\varphi}_\mathfrak{N}= \yvval{(\varphi^\bullet)'}_\mathfrak{M}=\yvval{(\varphi^\circ)''}_\mathfrak{M}$
\item for any sequent $\varphi\proves\psi$ in the language $\mathcal{L}$ of distribution-free modal logic  $\mathfrak{M}\models\varphi^\bullet\proves\psi^\bullet$ iff $\mathfrak{N}\models\varphi\proves\psi$ iff $\mathfrak{M}\vmodels\psi^\circ\vproves\varphi^\circ$,
\end{enumerate}
where $\mathfrak{N}$ is defined as above, by setting $\bar{V}^1(p_i)=V_1(P_i)''$.
\end{thm}
\begin{proof}
The translation and co-translation are special instances of the case of the languages and logics of arbitrary normal lattices expansions. A proof of all three claims for the general case was given in \cite[Theorem~3.2]{vb}.
\end{proof}

By the full abstraction Theorem~\ref{full abstraction of trans in sorted modal}, we may regard the language of modal lattices as a fragment $\mathcal{L}_\mathrm{reg}$ of $\mathcal{L}_1$, to which we refer as {\em the regular fragment}, and also as a fragment $\mathcal{L}_\mathrm{coreg}$ of $\mathcal{L}_\partial$, to which we may refer as {\em the co-regular fragment}. 

The standard translation of sorted modal logic into sorted  FOL  is exactly as in the single-sorted case, except for the relativization to two sorts, displayed in Table \ref{std-trans}, where $\stx{u}{}, \sty{v}{}$ are defined by mutual recursion and $u,v$ are individual variables of sort $1,\partial$, respectively.

\begin{table}[!htbp]
\caption{Standard Translation of the sorted modal language $\mathcal{L}_s=(\mathcal{L}_1,\mathcal{L}_\partial)$}
\label{std-trans}
($u$ a sort-1 variable, $v$ a sort-$\partial$ variable)
\begin{tabbing}
$\stx{u}{P_i}$\hskip0.8cm\==\hskip2mm\= ${\bf P}_i(u)$ \hskip3.5cm\= $\sty{v}{P^i}$\hskip0.8cm\==\hskip2mm\= ${\bf P}^i(v)$\\
$\stx{u}{\top}$ \>=\> $u=u$ \> $\stx{v}{\top}$\>=\> $v=v$\\
$\stx{u}{\bot}$ \>=\> $u\neq u$ \> $\sty{v}{\bot}$ \>=\> $v\neq v$\\
$\stx{u}{\alpha\cap \eta}$\>=\> $\stx{u}{\alpha}\cap\stx{u}{\eta}$\> $\sty{v}{\beta\cap \delta}$\>=\> $\sty{v}{\beta}\cap\sty{v}{\delta}$\\
$\stx{u}{\alpha\cup \eta}$\>=\> $\stx{u}{\alpha}\cup\stx{u}{\eta}$\> $\sty{v}{\beta\cup \delta}$\>=\> $\sty{v}{\beta}\cup\sty{v}{\delta}$\\
$\stx{u}{\beta'}$ \>=\> $\forall^\partial v\;(u{\bf I}v\;\lra\;\neg\sty{v}{\beta})$ \> $\sty{v}{\alpha'}$ \>=\> $\forall^1 u\;(u{\bf I}v\;\lra\;\neg\stx{u}{\alpha})$\\
$\stx{u}{\diamondvert\alpha}$ \>=\> $\exists^1 z\;(u{\bf R}^{11}_\Diamond z\;\wedge\;\stx{z}{\alpha})$ \> $\sty{v}{\diamondminus\beta}$ \>=\> $\exists^\partial y\;(v{\bf R}^{\partial\partial}_\Box y\;\wedge\;\stx{y}{\beta})$
\\
$\stx{u}{\bbox_1\alpha}$ \>=\> $\forall^1z(z\mathbf{R}^{11}_\Diamond u\lra\stx{z}{\alpha})$
\> 
$\sty{v}{\bbox_\partial\beta}$ \>=\> $\forall^\partial y(y\mathbf{R}^{\partial\partial}_\Box v\lra\sty{y}{\beta})$
\\
$\stx{u}{\boxminus\alpha}$ \>=\> $\forall^1 z\;(u{\bf R}''_\Box z\lra\stx{z}{\alpha})$\> $\sty{v}{\boxvert\beta}$ \>=\> $\forall^\partial y(v{\bf R}''_\Diamond y\lra\sty{y}{\beta})$
\\
$\stx{u}{\btdown\beta}$ \>=\>  $\forall^\partial y(x\mathbf{R}''_\tdown y\lra\sty{y}{\beta})$     
\> $\sty{v}{\tdown\alpha}$ \>=\> $\exists^1x(v\mathbf{R}^{\partial 1}_\tdown x\wedge\stx{x}{\alpha})$ 
\\
$\stx{u}{\alpha\odot\eta}$ \>=\>    $\exists^1x\exists^1 z(u\mathbf{R}^{111}xz\wedge\stx{x}{\alpha}\wedge\stx{z}{\eta})$
\\
\>\>\hskip3cm $\sty{v}{\alpha\tright\beta}$ \;=\; $\exists^1x\exists^\partial y(v\mathbf{T}^{\partial 1\partial}xy\wedge\stx{x}{\alpha}\wedge\sty{y}{\beta})$
\\
$\stx{u}{\alpha\rspoon\eta}$  \>=\> $\forall^1x\forall^1z(z\mathbf{R}^{111}xu\wedge\stx{x}{\alpha}\lra\stx{z}{\eta})$
\end{tabbing}
\hrule
\end{table}

If $\alpha(Q_{i_1},\ldots,Q_{i_n})$, where for each $j$, $Q_{i_j}\in\{P_{i_j}, P^{i_j}\}$, is an $\mathcal{L}_1$-sentence with propositional variables among the $Q_{i_j}$, then its second-order translation is defined to be the sentence $\mathrm{ST}^2_{x}(\alpha)=\forall Q_{i_1}\ldots\forall Q_{i_n}\forall^1 u\;\stx{u}{\alpha}$. It is understood that $\forall Q_{i_j}$ is $\forall^1 P_{i_j}$, if $Q_{i_j}=P_{i_j}\in\mathcal{L}_1$ and it is $\forall^\partial P^{i_j}$ otherwise. Similarly for $\beta(Q_{i_1},\ldots,Q_{i_n})$ and $\mathrm{ST}^2_v(\beta)=\forall Q_{i_1}\ldots\forall Q_{i_n}\forall^\partial v\sty{v}{\beta}$.

\begin{prop}
\label{std trans prop}
For any sorted modal formulae $\alpha,\beta$ (of sort $1, \partial$, respectively), any model $\mathfrak{M}=(\mathfrak{F},V)$ for $\mathcal{L}^1_s$  and any $x\in Z_1, y\in Z_\partial$, $\mathfrak{F}\models\mathrm{ST}^2_{u}(\alpha)[u:=x][Q_{i_j}:=V(Q_{i_j})]_{j=1}^n$ iff $\mathfrak{M},x\models \alpha$ iff $\mathfrak{F}\models\stx{u}{\alpha}[u:=x]$, where $x=V(u)$.

Similarly, $\mathfrak{M},y\vmodels \beta$ iff $\mathfrak{F}\vmodels\sty{v}{\beta}[v:=y]$ iff $\mathfrak{F}\vmodels\mathrm{ST}_v^2(\beta)[v:=y][Q_{i_j}:=V(Q_{i_j})]_{j=1}^n$.
\end{prop}

\begin{coro}
A sequent $\alpha\proves\eta$ in the sorted modal logic corresponds to the implication $\stx{u}{\alpha}\ra\stx{u}{\eta}$. In other words,  for any model $\mathfrak{M}=(\mathfrak{F},V)$, $\mathfrak{M}\models\alpha\proves\eta$ iff $\val{\alpha}_\mathfrak{M}\subseteq\val{\eta}_\mathfrak{M}$ iff $\mathfrak{F}\models(\stx{u}{\alpha}\ra\stx{u}{\eta})[V]$.
\end{coro}

\section{Sahlqvist - Van Benthem Correspondence}
\label{sahlqvist section}
In Section~\ref{reduction section} we introduce the reduction strategy, designed to reduce a sequent in the language of distribution-free modal logic with negation and implication (DfML) to what we define to be a system of inequalities in the reduction language (the language of the modal companion of DfML) in {\em canonical Sahlqvist form}. The section is concluded with Definition~\ref{sahlqvist sequents and inequality systems}, fixing our understanding of when a sequent $\varphi\proves\psi$ in the language of DfML is Sahlqvist. Section~\ref{reduction structure section} describes the main steps of the (non-deterministic) generalized Sahlqvist -- Van Benthem algorithm, introducing {\em $\mathrm{t}$-invariance constraints} and the {\em guarded second-order translation}. To simplify the presentation we first prove, in Section~\ref{box and diamond section}, the correspondence result for the fragment of the language with necessity and possibility only, extending to negation in Section~\ref{negation section} and to implication in Section~\ref{implication section}.

\subsection{Reduction Rules and Sahlqvist Sequents}
\label{reduction section}
A {\em positive occurrence} of a propositional variable in a sentence $\zeta$ in the language of sorted modal logic without implication $\rspoon$ is one in the scope of an even number of applications of the priming operator. The variable {\em occurs positively} in $\zeta$ iff every one of its occurrences is positive. A sentence $\zeta$ is {\em positive} iff every propositional variable that occurs in $\zeta$ occurs positively in it.

\begin{defn}[Simple Sahlqvist Sequents]
\label{simple sahlqvist}
A {\em simple Sahlqvist sequent} $\alpha\proves\eta$ of the first sort, or $\beta\vproves\delta$ of the second sort, is a sequent with positive consequent $\eta$, respectively $\delta$, and such that the premiss $\alpha$, respectively $\beta$, of the sequent is built from $\top, \bot$ and boxed atoms by closing under conjunction and the additive operators $\diamondvert, \odot$ for the first sort, and $\diamondminus,\tdown$ and $\tright$ for the second sort. 
\end{defn}

Examples of simple Sahlqvist sequents are $\boxminus P\proves P, {\boxvert\boxvert} Q\vproves Q$, $\boxminus P\proves\diamondvert P$, ${\diamondvert\diamondvert}P\proves{\diamondvert} P$, $P\tright Q\proves P\tright(P\tright Q)$, $P\vproves \tdown Q$, but  $\boxminus P''\proves P''$ (i.e. the sequent $\bb p\proves p$ in the regular fragment of the language) is not simple Sahlqvist. Proving that every simple Sahlqvist sequent in the language of sorted modal logic effectively locally corresponds to a first-order formula can be done in the same way as in the classical case and no need for separate proof arises.

Semantically, a sequent can be equivalently regarded as a formal inequality $\zeta\leq_\sharp \xi$, where if the sequent is $\zeta\proves \xi$, then $\sharp=1$ and if the sequent is $\zeta\vproves \xi$, then $\sharp=\partial$. 

Simple Sahlqvist inequalities are defined in the obvious way, given Definition~\ref{simple sahlqvist}. 

Most of the pre-processing in the generalized Sahlqvist - Van Benthem correspondence algorithm consists in manipulating (reducing) formal systems of inequalities, which are systems of the following form, where  $n,m\geq 0, \mbox{ and } \sharp,\sharp_i\in\{1,\partial\}$,
\begin{equation}\label{set of inequalities}
S=\lset  Q_1''\leq_{\sharp_1} Q_1,\ldots,Q_n''\leq_{\sharp_n} Q_n, Q_{n+1}=_{\sharp_{P'_1}}P'_1,\ldots,Q_{n+m}=_{\sharp_{P'_m}}P'_m \midsp \zeta\leq_\sharp \xi \rset, 
\end{equation}
and where $\zeta\leq_\sharp \xi $ is its {\em main inequality} and $Q_i$ are propositional variables of sort determined by $\sharp_i$, for each $i=1,\ldots,n$ and $Q_{n+i}$ are of the same sort as $P_i'$, as indicated by the subscript to the equality symbol. 

Reduction aims at eliminating any occurrence of the priming operator on the left-side of the main inequality, as well as any occurrence of implication $\rspoon$.

We refer to formal inequalities of the form $Q''\leq_{\sharp Q}Q$ as {\em stability constraints} and to formal equations of the form $Q=_{\sharp_{P'}}P'$ as {\em change-of-variables constraints}. For brevity, we  write $\lset\mathrm{STB,CVC}\midsp\zeta\leq_\sharp\xi\rset $, at times displaying some constraints of interest included in $\mathrm{STB}$ and/or in $\mathrm{CVC}$.

\begin{defn}\label{equivalence of sets of inequalities}
For sets of formal inequalities $S_1,S_2$, define an equivalence relation by $S_1\sim S_2$ iff for any model $\mathfrak{M}=(\mathfrak{F},V)$ satisfying all the constraints, of the form $Q''\leq_{\sharp_Q} Q$, or $Q=_{\sharp_{P'}}P'$, in each of $S_1,S_2$, the model validates the main inequality of $S_1$ iff it validates the main inequality of $S_2$.
\end{defn}

\begin{ex}
$\lset\bb p\leq p\rset$,
$\lset\boxminus P''\leq_1P''\rset$, $\lset P''\leq_1P\midsp \boxminus P\leq_1 P\rset$, $\lset Q''\leq_1Q,P''\leq_1P\midsp \boxminus P\leq_1P\rset$ are equivalent in the sense of Definition~\ref{equivalence of sets of inequalities}. The first is an inequality in the regular fragment of the sorted modal language, which is then just an abbreviation for the second inequality $\boxminus P''\leq_1P''$. For the second and third sets,  $V(P)''\subseteq V(P)$ means that $V(P)=A$ is a Galois stable set, hence $A=A''$. The main (only) inequality of the second set is valid in a model iff $\lbminus A''\subseteq A''$ iff $\lbminus A\subseteq A$ iff the main inequality of the third set is valid in $\mathfrak{M}$. In the last set, the addition (or removal) of a redundant constraint $Q''\leq_1 Q$ is not relevant in the evaluation of the main inequalities.
\end{ex}

Table~\ref{reduction rules} presents a set of effectively executable reduction rules for sets of inequalities, proven to preserve equivalence in Lemma~\ref{R1-R6}. Note that  rule (R3) is a special instance of (R2), for $n=1$ and $\xi_1=\xi'$, but we include it because of its usefulness.
\begin{table}[!htbp]
\caption{Reduction Rules}
\label{reduction rules}
\begin{enumerate}
\item[(R1)] $\infrule{\lset \mathrm{STB,CVC}, P''\leq_{\sharp_P}P\midsp \zeta\leq_\sharp\xi\rset}{\lset \mathrm{STB,CVC}\midsp \zeta\leq_\sharp\xi\rset}$,\\ provided the propositional variable $P$ does  not occur in $\zeta$ or $\xi$
    \\
\item[(R2)]  $\infrule{\lset \mathrm{STB,CVC}\midsp \zeta''\leq_\sharp\xi_1'\cap\cdots\cap\xi'_n\rset}{\lset \mathrm{STB,CVC}\midsp \zeta\leq_\sharp\xi_1'\cap\cdots\cap\xi_n'\rset}$, for $n\geq 1$
\\
\item[(R3)]  $\infrule{\lset \mathrm{STB,CVC}\midsp \zeta''\leq_\sharp\xi''\rset}{\lset \mathrm{STB,CVC}\midsp \zeta\leq_\sharp\xi''\rset }$
\\
\item[(R4)] $\infrule{\lset \mathrm{STB,CVC}\midsp \zeta\leq_\sharp\xi\rset}{\lset \mathrm{STB,CVC}, P''\leq_{\sharp_P}P\midsp \zeta[P/P'']\leq_\sharp\xi[P/P'']\rset}$,\\
provided every occurrence of the propositional variable $P$ in each of $\zeta,\xi$ is double-primed and where
     $\sharp_P$ is the sort of $P$ and $\zeta[P/P''],\xi[P/P'']$ designate the results of uniformly replacing each occurrence of $P''$ by one of $P$ in each of $\zeta,\xi$
\\
\item[(R5)] If a re-write rule from the list $\mathrm{REWRITE}$ in Table~\ref{rewrite rules table} is applicable, update the system of inequalities by carrying out the re-write
\\
\item[(R6)] $\infrule{\lset \mathrm{STB,CVC}\midsp\zeta\leq_\sharp\xi\rset}{\lset \mathrm{STB,CVC},Q=_{\sharp_{P'}}P'\midsp\zeta[Q/P']\leq_\sharp\xi[Q/P']\rset }$,\\
provided the variable $P$ occurs in the main inequality only single-primed and $Q$ is a fresh variable of the same sort as $P'$.
\item[(R7)] Apply residuation to rewrite the main inequality according to the related rewrite rule
    \begin{tabbing}
    $\alpha\leq_1\eta\rspoon\zeta$ \hskip1cm\=$\mapsto$\hskip4mm\= $\eta\odot\alpha\leq_1\zeta$
    \hskip3cm\= $\diamondvert\alpha''\leq_1\eta$ \hskip1cm\=$\mapsto$\hskip4mm\= $\alpha''\leq_1\bbox_1\eta$
\\
 \>\> \> $\diamondminus\beta''\leq_\partial\delta$ \> $\mapsto$\> $\beta''\leq_\partial\bbox_\partial\delta$
    \end{tabbing}
\item[(R8)] $\infrule{\lset \mathrm{STB,CVC}\midsp \zeta\rspoon P\leq_\sharp\xi\rspoon P\rset}{\lset \mathrm{STB,CVC}\midsp \xi\leq_\sharp\zeta\rset}$\\
\item[(R9)] $\infrule{\lset \mathrm{STB,CVC}\midsp \zeta''\cap\boxplus P\leq_\sharp \theta_1'\cap\cdots\cap\theta_n'\rset}{\lset \mathrm{STB,CVC}\midsp \zeta\cap\boxplus P\leq_\sharp \theta_1'\cap\cdots\cap\theta_n'\rset}$,
      $\boxplus\in\{\boxminus,\boxvert\}$, $P$  constrained in $\mathrm{STB,CVC}$, $n\geq 1$
\end{enumerate}
\hrule
\end{table}

Introduction of change-of-variables constraints (rule (R6) in Table~\ref{reduction rules}) is motivated in the following example.
\begin{ex}\label{cvc example}
Consider the system $\lset P''\leq_1P\midsp(\diamondvert P)'\leq_\partial P'\rset $. 
By (R5.2), Table~\ref{rewrite rules table}, the system is equivalent to $\lset P''\leq_1P\midsp {\boxvert} P'\leq_\partial P'\rset$. As $P$ does not occur except single-primed, $P'$ designates an arbitrary Galois co-stable set $B$ and the inequality is valid iff $\lbvert B\subseteq B$. The value $V(P)\subseteq Z_1$ is not relevant in the evaluation of the inequality, just because $P$ does not occur unprimed. Consider $\lset P''\leq_1P, Q=_\partial P'\midsp \lbvert Q\leq_\partial Q\rset $. The change-of-variables constraint $Q=_\partial P'$ allows us to conclude that the two systems $\lset P''\leq_1P,Q=_\partial P'\midsp \lbvert Q\leq_\partial Q\rset $ and $\lset P''\leq_1P\midsp \boxvert P'\leq_\partial P'\rset $ are equivalent. As $P$ no longer occurs in the main inequality, by (R1) the system reduces to $\lset Q=_\partial P'\midsp \lbvert Q\leq_\partial Q\rset $.
\end{ex}

Replacing subterms by equivalent ones in an inequality $\zeta\leq_\sharp\xi$ obviously leads to equivalent systems of inequalities.  A minimal useful finite list of re-write rules is displayed in Table~\ref{rewrite rules table}. All rules are uni-directional, to avoid loops in execution. Most of the cases relate to taking advantage of the result in Proposition~\ref{box as restriction}, namely that the restriction to Galois sets of the box operator generated by a double-dual relation ($R''_\Box$, or $R''_\Diamond$)  is the dual to the diamond operator of the other sort, i.e. $\lbminus A=(\ldminus A')'$ and $\lbvert B=(\ldvert B')'$, for $A\in\gpsi$ and $B\in\gphi$.

\begin{table}[!htbp]
  \caption{REWRITE Rules}\label{rewrite rules table}
  (``$P$ constrained in $\mathrm{STB,CVC}$'' means that either $P''\leq_{\sharp_P}P$ is in $\mathrm{STB}$, or $P=_{\sharp_{Q'}}Q'$ is in $\mathrm{CVC}$)
  \begin{tabbing}
(R5.1) \hskip2mm\= $(\diamondminus\alpha')'$\hskip1.5cm\= $\mapsto\boxminus\alpha''$ \hskip1.8cm\= $\alpha\in\mathcal{L}_1$
\\
       \> $(\diamondvert\beta')'$\>$\mapsto\boxvert\beta''$ \>  $\beta\in\mathcal{L}_\partial$
\\[2mm]
(R5.2) \> $(\diamondminus P)'$\>$\mapsto \boxminus P'$ \>  $P$ constrained in $\mathrm{STB,CVC}$
\\
\> $(\diamondvert P)'$\>$ \mapsto \boxvert P'$  \> $P$ constrained in $\mathrm{STB,CVC}$
\\[2mm]
(R5.3) \> $(\boxminus P)''$\>$\mapsto\boxminus P$ \> $P$ constrained in $\mathrm{STB,CVC}$
\\
\> $(\boxvert P)''$\>$\mapsto \boxvert P$ \> $P$ constrained in $\mathrm{STB,CVC}$
\\[2mm]
(R5.4) \>  $P'''$\>$\mapsto P'$
\\[2mm]
(R5.5) \> $\boxminus(\zeta_1\cap\zeta_2)$\>$\mapsto \boxminus\zeta_1\cap\boxminus\zeta_2$ \> $\zeta_1,\zeta_2\in\mathcal{L}_1$
\\
\> $\boxvert(\zeta_1\cap\zeta_2)$\>$\mapsto \boxvert\zeta_1\cap\boxvert\zeta_2$ \> $\zeta_1,\zeta_2\in\mathcal{L}_\partial$
\\[2mm]
(R5.6) \> $(\eta\cap\zeta)''$\>$\mapsto   \eta''\cap\zeta''$
\\
\> $(\eta\cup\zeta)'$\>$\mapsto\eta'\cap\zeta'$
\\
(R5.7) 
\> $(\tdown\alpha'')'$\> $\mapsto$  $\boxtimes\alpha'$
\\
\> $(\tdown P)'$\>$\mapsto\boxtimes P'$ \> $P$ constrained in $\mathrm{STAB, CVC}$
\\
(R5.8)\> $P\rspoon Q$\>$\mapsto (P\tright Q')'$  \> $P,Q$ constrained in $\mathrm{STB}$, $\mathrm{CVC}$
\\[2mm]
(R5.9) \> $P_2\rspoon(P_1\rspoon Q)$\>$\mapsto P_1\odot P_2\rspoon Q$ \> for any variables $P_1,P_2,Q$
\end{tabbing}
\hrule
\end{table}

\begin{lemma}\label{R1-R6}
Executing any of the actions (R1)--(R9) listed in Table~\ref{reduction rules} to a system $\mathcal{S}_1$ of inequalities leads to an equivalent system $\mathcal{S}_2$.  
\end{lemma}
\begin{proof}
We treat each case in turn.
\begin{enumerate}
\item[(R1)] Immediate, exactly because the variable $P$ does not occur in $\zeta\leq_\sharp\xi$. 
\item[(R2)] Straightforward, following from the fact that if $G$ is a Galois set and $W$ is any set, then $W''\subseteq G$ iff $W\subseteq G$.
\item[(R3)] The argument is the same as for (R2).
\item[(R4)]  For the introduction of a stability constraint rule (R4), validity of the inequality depends only on interpreting $P''$ as a Galois set and since $P$ does not occur unprimed, too, by assumption, the particular value of $V(P)$ is not relevant.
\item[(R5)] Immediate, since the rule just performs replacement of equivalents.
\item[(R6)] Straightforward, since by Definition~\ref{equivalence of sets of inequalities} the formal equality $Q=_{\sharp_{P'}}P'$ must be also valid.
\item[(R7)] Immediate.
\item[(R8)] From $\xi\leq_\sharp\zeta$ it is immediate that we obtain $\zeta\rspoon P\leq_\sharp\xi\rspoon\zeta$. The converse is straightforward, too, since in particular by substitution we obtain $\zeta\rspoon\zeta\leq_\sharp \xi\rspoon\zeta$. The latter is equivalent to each of $\top\leq_\sharp\xi\rspoon\zeta$ and $\xi\leq_\sharp\zeta$.
\item[(R9)] The rule was written in this form to avoid loops that can otherwise arise given also the re-write rules (R5.3) and (R5.6). The constraint on $P$ means that $P$ is interpreted as a Galois set $G=V(P)$.  Letting $\val{\zeta''}=U''$ and $\val{\theta_1'\cap\cdots\cap\theta_n'}=G_1$ be the Galois set interpreting the conclusion, and given that $\lbplus G=(\lbplus G)''$ for $\lbplus\in\{\lbminus,\lbvert\}$, by Proposition~\ref{box as restriction} and Corollary~\ref{boxvert as restriction}, we have $U''\cap\lbplus G=(U\cap\lbplus G)''\subseteq G_1$ iff $U\cap\lbplus G\subseteq G_1$, since for any set $W$, $W\subseteq G_1$ iff $W''\subseteq G_1$, given that $G_1$ is a Galois set.
\end{enumerate}
All cases have been examined, hence the proof is complete.
\end{proof}

\begin{defn}
  \label{canonical Sahlqvist form}
A system $\lset \mathrm{STB,CVC}\midsp\zeta\leq_\sharp\xi\rset $ is in {\em canonical Sahlqvist form} if the main inequality $\zeta\leq_\sharp\xi$ is simple Sahlqvist and for any stability constraint $P''_1\leq_{\sharp_{P_1}}P_1$ in $\mathrm{STB}$ and change-of-variables constraint $P_2=_{\sharp_{Q'}}Q'$ in $\mathrm{CVC}$, $P_1$ and $P_2$ occur only unprimed in $\zeta,\xi$.
\end{defn}
Note that in the right-hand-side of the inequality an unprimed variable $P$ may  be within the scope of $(\;)'$, if $P$ occurs as a subterm of a primed term, as in $(\diamondvert P)''$.

\begin{defn}\label{sahlqvist sequents and inequality systems}
A system of inequalities as in \eqref{set of inequalities} {\em is Sahlqvist} if it can be reduced to canonical Sahlqvist form, using the reduction rules (R1)--(R9). 

A 1-sequent $\zeta\proves\xi$ is Sahlqvist if  the associated inequality system $\lset \zeta\leq_1\xi\rset $ is Sahlqvist. Similarly for a $\partial$-sequent $\zeta\vproves\xi$.

A sequent $\varphi\proves\psi$ in the language of distribution-free modal logic with negation and implication is Sahlqvist iff either its translation $\varphi^\bullet\proves\psi^\bullet$, or its co-translation (dual translation) $\psi^\circ\vproves\varphi^\circ$ is Sahlqvist.
\end{defn}

\begin{lemma}\label{decidable sahlqvist}
  It is decidable whether a sequent $\varphi\proves\psi$ in the language of DfML is Sahlqvist, or not.
\end{lemma}
\begin{proof}
Given $\varphi\proves\psi$ as input, non-deterministically choose one of $\varphi^\bullet\leq_1\psi^\bullet$ or $\psi^\circ\leq_\partial\varphi^\circ$. For each possible choice, if multiple equivalent ways to (co)translate a subsentence exist, make a non-deterministic choice for one of them. Having completed the (co)translation and as long as some reduction rule applies, make a non-deterministic choice among the applicable rules and update the system of inequalities. The statement of rules guarantees that there can be no loops in rule application, hence the process eventually terminates.  By inspection of syntax it is decidable whether the output of this process is in canonical Sahlqvist form. This decides whether the input sequent is Sahlqvist.
\end{proof}

We provide below some examples, showing that there are cases where both the translation and the dual translation can be reduced to canonical Sahlqvist form, but there are also cases where only one can be so reduced (and there are of course cases where none can be so reduced).

\begin{ex}
  \label{both}
  The sequent $\bb p\proves p$ in the language of DfML translates to $\boxminus P''\proves P''$, which can be reduced quickly to $\lset P''\leq_1 P\midsp \boxminus P\leq_1 P\rset$. Its co-translation, according to Table~\ref{syntactic translation into sorted}, is the $\partial$-sequent $P'\vproves (\diamondminus P')''$, reducing to $\lset Q=_\partial P'\midsp Q\leq_\partial(\diamondminus Q)''\rset$. In this case both the translation and the co-translation reduce to canonical Sahlqvist form.

Consider also the DfML sequent $p\proves\dd p$. It  translates to the 1-sequent $P''\proves(\diamondvert P'')''$ and it reduces to $\lset P''\leq_1 P\midsp P\leq_1(\diamondvert P)''$, which is in canonical Sahlqvist form. 
  
Its co-translation is the $\partial$-sequent $(\diamondvert P'')'\vproves P'$. Applying (R5.1) with $\alpha=P'$ to the original system $\lset (\diamondvert P'')'\leq_\partial P'\rset$ we obtain its equivalent $\lset \boxvert P'''\leq_\partial P'\rset$ and then, by (R5.4), $\lset \boxvert P'\leq_\partial P'\rset$. By the change of variables rule (R6) this is further reduced to $\lset Q=_\partial P'\midsp \boxvert Q\leq_\partial Q\rset$, which is also in canonical Sahlqvist form. 
\end{ex}

\begin{ex}\label{only one is sahlqvist}
Consider the sequent ${\dd\dd} p\proves\dd p$ in the language of distribution-free modal logic. Recall that $(\dd p)^\bullet=(\diamondvert p^\bullet)''=(\diamondvert P'')''$ and then $({\dd\dd}p)^\bullet=(\diamondvert(\dd p)^\bullet)''=(\diamondvert(\diamondvert P'')'')''$

The translation of ${\dd\dd}p\proves{\dd}p$ is the 1-sequent $(\diamondvert(\diamondvert P'')'')''\proves(\diamondvert P'')''$ and its co-translation (dual translation) is the $\partial$-sequent $(\diamondvert P'')'\vproves (\diamondvert(\diamondvert P'')'')'$. The reduction sequence for the dual translation is shown below, reaching canonical Sahlqvist form.

\begin{tabbing}
1.\hskip6mm\=$\lset(\diamondvert P'')'\leq_\partial (\diamondvert(\diamondvert P'')'')'\rset$\hskip3cm\= \\
2.\>$\lset P''\leq_1P\midsp (\diamondvert P)'\leq_\partial (\diamondvert(\diamondvert P)'')'\rset$\> by (R4)\\
3.\>$\lset P''\leq_1P\midsp \boxvert P'\leq_\partial (\diamondvert(\boxvert P')')'\rset$\> by (R5.2)\\
4.\>$\lset P''\leq_1P,Q=_\partial P'\midsp \boxvert Q\leq_\partial (\diamondvert(\boxvert Q)')'\rset$\> by (R6)\\
5.\>$\lset Q=_\partial P'\midsp \boxvert Q\leq_\partial (\diamondvert(\boxvert Q)')'\rset$\> by (R1)
\end{tabbing}
The reduction rules applied to the 1-sequent do not succeed to reduce it to a system in canonical Sahlqvist form.

\begin{tabbing}
1.\hskip6mm\= $\lset (\diamondvert(\diamondvert P'')'')''\leq_1(\diamondvert P'')''\rset$\hskip2cm\=\\
2.\> $\lset P''\leq_1P\midsp (\diamondvert(\diamondvert P)'')''\leq_1(\diamondvert P)''\rset$\> by (R4)\\
3.\> $\lset P''\leq_1P\midsp (\diamondvert(\boxvert P')')''\leq_1(\diamondvert P)''\rset$\> by (R5.2)\\
4.\> $\lset P''\leq_1P\midsp (  \boxvert(\boxvert P')''   )''\leq_1(\diamondvert P)''\rset$  \> by (R5.1), with $\alpha=\boxvert P'$
\end{tabbing}
There are some further steps that can be carried out, but no further reduction rule application will result in eliminating all the occurrences of the priming operator on the left-hand-side of the inequality, as the reader can surely verify.
\end{ex}

\subsection{Structure of the Algorithm}
\label{reduction structure section}
\paragraph{Step 1 (Reduce to Canonical Sahlqvist Form).}
\noindent
\underline{Input:} A sequent $\varphi\proves\psi$ in the language of DfML.\\
Non-deterministically choose (spawn parallel threads) to process
 either the translation \mbox{$\varphi^\bullet\le_1\psi^\bullet$}, or the co-translation (dual translation) $\psi^\circ\leq_\partial\varphi^\circ$ of the input. Run the reduction process described in the proof of Lemma~\ref{decidable sahlqvist}. If neither of the (co)translation sequents reduces to a system of formal inequalities in canonical Sahlqvist form, then fail, else continue to step 2, with input either a system $\lset \mathrm{STB,CVC}\midsp\alpha\leq_1\eta\rset $, or a system $\lset \mathrm{STB,CVC}\midsp\beta\leq_\partial\delta\rset$, whichever was the output of this step.

\paragraph{Step 2 (Calculate $\mathrm{t}$-Invariance Conditions).}
\noindent
\underline{Input:} A system $\lset \mathrm{STB,CVC}\midsp\alpha\leq_1\eta\rset $ (or $\lset \mathrm{STB,CVC}\midsp\beta\leq_\partial\delta\rset$) in canonical Sahlqvist form, where $\mathrm{STB}=\{P''_i\leq_{\sharp_i}P_i\midsp i=1,\ldots,n\}$ and $\mathrm{CVC}=\{P_{n+j}=_{\sharp_{Q'_j}}Q'_j\midsp j=1,\ldots,k\}$.\\
Set $\mathrm{t{-}INV}=\top$.\\
{\bf For each} $i=1,\ldots,n+k$\\
\hspace*{5mm} {\bf If} $i\leq n$ and $\sharp_i=1$, or $n< i=n+j\leq n+k$ and $\sharp_{Q'_j}=1$\\
\hspace*{1cm} Update $\mathrm{t{-}INV}\leftarrow\mathrm{t{-}INV}\wedge\forall^1 u_i[\mathrm{t}(P_i)(u_i)\lra P_i(u_i)]$, \\
\hspace*{5mm} {\bf else} (cases $\sharp_i=\partial, \sharp_{Q'_j}=\partial$)\\
\hspace*{1cm} Update $\mathrm{t{-}INV}\leftarrow\mathrm{t{-}INV}\wedge\forall^\partial u_i[\mathrm{t}(P_i)(u_i)\lra P_i(u_i)]$,\\
where the introduced conjunct abbreviates the formula in expression~\eqref{t-invariance condition 1} in the ``if'' case
\begin{equation}\label{t-invariance condition 1}
\forall^1 u[\underbrace{\forall^\partial v(u\mathbf{I}v\lra\exists^1 u_1(u_1\mathbf{I}v\wedge {P_i}(u_1)))}_{t({P_i})(u)}\lra{P_i}(u)],
\end{equation}
and that in expression~\eqref{t-invariance condition dual} in the ``else'' case
\begin{equation}\label{t-invariance condition dual}
\forall^\partial y[\underbrace{\forall^1 z (z\mathbf{I}y\lra\exists^\partial v(z\mathbf{I}v\wedge{P_i}(v)))}_{\mathrm{t}({P_i}(u))}\lra {P_i}(y)],
\end{equation}
respectively. The terminology ``$\mathrm{t}${-}invariance'' was introduced in \cite{vb}.

\paragraph{Step 3 (Generate the Guarded Second-Order Translation).}
\noindent
\underline{Input:} A guard $\mathrm{t{-}INV}=\bigwedge_{i=1}^{n+k}\forall^{\sharp_i} u_i[\mathrm{t}(P_i)(u_i)\lra P_i(u_i)]$, where for each $i=1,\ldots,n+k$, $\sharp_i\in\{1,\partial\}$ is the sort of $P_i$ and $\forall^{\sharp_i}\in\{\forall^1,\forall^\partial\}$, according to the value of $\sharp_i$.\\
\underline{Output:} The  {\em guarded second-order translation}, an expression of the form
\begin{equation}\label{guarded second order}
\forall^{\sharp_1} P_1\cdots\forall^{\sharp_{n+k}} P_{n+k}\forall^{\flat_1} Q^*_1\cdots\forall^{\flat_m} Q^*_m\forall^1 x(\mathrm{t{-}INV}\wedge\;\mathrm{ST}_x(\alpha)\lra\mathrm{ST}_x(\eta))
\end{equation}
or of the form 
\begin{equation}\label{guarded second order dual}
\forall^{\sharp_1} P_1\cdots\forall^{\sharp_{n+k}} P_{n+k}\forall^{\flat_1} Q^*_1\cdots\forall^{\flat_m} Q^*_m\forall^\partial y(\mathrm{t{-}INV}\wedge\;\mathrm{ST}_y(\beta)\lra\mathrm{ST}_y(\delta))
\end{equation}
depending on whether the translation, or the co-translation is being processed and, furthermore (a)  $P_i$, for $i=1,\ldots,n+k$, and $Q^*_j$, for $j=1,\ldots,m$, are all the second-order variables occurring in $\mathrm{ST}_x(\alpha)$ and $\mathrm{ST}_x(\eta)$ (respectively, in $\mathrm{ST}_y(\beta)$ and $\mathrm{ST}_y(\delta)$), corresponding to the propositional variables $P_i,Q^*_j$ occurring in $\alpha,\eta$ (respectively, in $\beta,\delta$) and (b) $\sharp_i\in\{1,\partial\}$ and $\flat_j\in\{1,\partial\}$ designate the sort of $P_i$ and of $Q^*_j$, respectively.

\paragraph{Step 4 (Pull-out Existential Quantifiers).}
This step is the same as in the classical case, using familiar equivalences to pull existential quantifiers in prenex position. It is detailed in the course of the proof of the correspondence result (Theorem~\ref{Sahlqvist thm}).

\paragraph{Step 5 (Determine Minimal Instantiations).}
The formal detail is explained in the course of the proof of Theorem~\ref{Sahlqvist thm}. We clarify the issues involved by discussing an example.

Consider the sequent $p\proves\dd p$ in the language of distribution-free modal logic. Applying the translation of Table~\ref{syntactic translation into sorted} we obtain the 1-sequent $P''\proves(\diamondvert P'')''$. By reduction, we obtain the system
 $\lset P''\leq_1 P\midsp P\leq_1 (\diamondvert P)''\rset $ and the guarded second-order translation is
\begin{equation}\label{so2}
\forall^1 P\forall^1 x[\forall^1 u[\mathrm{t}(P)(u)\lra P(u)]\wedge P(x)\lra\mathrm{ST}_x((\diamondvert P)''), \mbox{ where }
\end{equation}
  \begin{equation}\label{consequent}
 \mathrm{ST}_x((\diamondvert P)'')= \forall^\partial v[x\mathbf{I}v\lra\exists^1 z(z\mathbf{I}v\wedge\exists^1 u(z\mathbf{R}^{11}_\Diamond u\wedge P(u)))].
  \end{equation}
For the minimal instantiation in the classical case we set $\lambda(P)=\lambda s.(s=x)$. This is because we interpret $P$ as the smallest set possible and since $P(x)$ is in the antecedent (which we think of as $x\in P$), we choose $\{x\}$ as the interpretation of $P$.
But $\{x\}$ is not a stable set and this interpretation results in falsifying the $\mathrm{t}${-}invariance condition, hence making the implication in~\eqref{so2} vacuously true (so that no correspondent is computed). Instead, the minimal stable set containing $x$ is the principal upper set $\Gamma x=\{z\midsp x\leq z\}$. If $P$ is (interpreted as) a principal upper set, i.e. a closed element $\Gamma x$, then we need to set $\lambda(P)=\lambda s.x\leq s$, so that, allowing ourselves some notational abuse, $u\in P$ iff $x\leq u$, which is what we obtain by the $\beta$-reduction $\lambda(P)u=(x\leq u)$.

\paragraph{Step 6 (Eliminate Second-Order Quantifiers).}
The rationale is the same as in the classical Sahlqvist -- Van Benthem algorithm, substituting $\lambda(P)$ for $P$ and performing $\beta$-reduction, and it is detailed in the proof of Theorem~\ref{Sahlqvist thm}.

\subsection{Box and Diamond}
\label{box and diamond section}
To simplify the presentation, we first restrict to the fragment $\mathcal{L}_{\Box\Diamond}$ of the language of distribution-free modal logic, returning to a treatment of negation and implication in the subsequent two Sections~\ref{negation section} and~\ref{implication section}.

\begin{thm}\label{Sahlqvist thm}
  Every Sahlqvist sequent in the language $\mathcal{L}_{\Box\Diamond}$  of distribution-free modal logic restricted to $\bb,\dd$ has a first-order local correspondent, effectively computable from the input sequent.
\end{thm}
\begin{proof}
Note that if $\lset\mathrm{STB,CVC}\midsp\zeta\leq_\sharp\xi\rset$ is a system of inequalities in canonical Sahlqvist form obtained by pre-processing (reducing) a sequent $\varphi\proves\psi$ in the fragment $\mathcal{L}_{\Box\Diamond}$ of the language of DfML, then for each stability constraint $P''\leq_{\sharp_P}P$ and each change-of-variables constraint $P=_{\sharp_{Q'}}Q'$ we shall have $\sharp_P=\sharp=\sharp_{Q'}$. This does not affect the proof argument, but it does simplify notation in the guarded second-order translation.

\paragraph{Reduction and Guarded Second-Order Translation.}
The hypothesis implies that pre-processing terminates with output a system of formal inequalities in canonical Sahlqvist form.
 
We give the proof for the case $\lset\mathrm{STB,CVC}\midsp \alpha\leq_1\eta\rset$, as the case where the co-translation was processed is similar.

Given $\lset\mathrm{STB,CVC}\midsp \alpha \leq_1\eta \rset$ in canonical Sahlqvist form,  let \eqref{guarded second order}, repeated below, 
\[
\forall^1 P_1\cdots\forall^1 P_n\forall^1 P_{n+1}\cdots\forall^1 P_{n+k}\forall^1 Q^*_1\cdots\forall^1 Q^*_m\forall^1 x(\mathrm{t{-}INV}\wedge\;\mathrm{ST}_x(\alpha )\lra\mathrm{ST}_x(\eta ))
\]
be its guarded second order translation. By a change of bound variables, if necessary, we ensure that quantifiers bind distinct variables.

\paragraph{Special/Trivial Cases.}
If a quantifier binds a second-order variable $P$ that occurs only in the consequent $\eta $, then replace such occurrences of $P$ by $\bot$, remove the quantification $\forall^1P$ and, if applicable, remove the $\mathrm{t}$-invariance conjunct $\forall^1u[\mathrm{t}(P)(u)\ra P(u)]$.

If $\mathrm{ST}_x(\alpha )$ is equivalent to $\top$, then the formula $\mathrm{ST}_x(\tilde{\eta})$ is a local first-order correspondent, where $\mathrm{ST}_x(\tilde{\eta})$ results by executing the action in the previous paragraph.

If $\mathrm{ST}_x(\alpha )$ is equivalent to $\bot$, then \eqref{guarded second order} is equivalent to $\forall^1 x(\bot\ra\mathrm{ST}_x(\tilde{\eta}))$, where $\mathrm{ST}_x(\tilde{\eta})$ is as in the previous paragraph, and $x=x$ may be taken to be a trivial local first-order correspondent.

\paragraph{Pulling-out Existential Quantifiers.}
Use equivalences of the form $\exists^1 z(\zeta_1(z)\wedge \zeta_2)\equiv \exists^1z\zeta_1(z)\wedge\zeta_2$ and $\forall^1z(\zeta_1(z)\lra\zeta_2)\equiv \exists^1z\zeta_1(z)\lra\zeta_2$ to pull in prenex position all occurring existential quantifiers in $\mathrm{ST}_x(\alpha )$. This modifies~\eqref{guarded second order} to the form shown in~\eqref{step1-modified so}, 
\begin{equation}\label{step1-modified so}
\forall^1 P_1\cdots\forall^1P_{n+k}\forall^1Q^*_1\cdots\forall^1Q^*_m\forall^1 x\forall^1z_1\cdots\forall^1 z_r(\mathrm{t{-}INV}\wedge\mathrm{REL}\wedge\mathrm{AT}\wedge\mathrm{BOXED}\lra\mathrm{POS}),
\end{equation}
where 
\begin{itemize}
\item $P_i$, for $i=1,\ldots,n+k$, and $Q^*_j$, for $j=1,\ldots,m$, are all the second-order variables occurring in $\mathrm{ST}_x(\alpha )$ and $\mathrm{ST}_x(\eta )$, corresponding to the propositional variables $P_i,Q^*_j$ occurring in $\alpha ,\eta $, where for $i=1,\ldots,n$ a stability constraint $P''_i\leq_1 P_i$  for the interpretation is in $\mathrm{STB}$, while for $i=n+1,\ldots,n+k$ a change of variables constraint $P_i=_1Q'_i$ is in $\mathrm{CVC}$,
\item $\mathrm{POS}=\mathrm{ST}_x(\eta )$, 
\item $\mathrm{t{-}INV}$ 
is the conjunction of the $\mathrm{t}$-invariance constraints for the second-order variables $P_i, i=1,\ldots,n+k$, 
\item $\mathrm{REL}$
is a conjunction of relational atomic formulae 
\item $\mathrm{AT}$
is a conjunction of atomic formula of the form $P(u)$ 
\item $\mathrm{BOXED}$
is a conjunction of formulae of the form $\forall^1u(w(\mathbf{R}''_\Box)^{k+1}u\lra P(u))$, corresponding to the translation of sentences of the form $\bb^{k+1}P$, for some $k\geq 0$.
\end{itemize}
\paragraph{Determining Minimal Instantiations.}
\subparagraph{Unconstrained Variables:}
For each $Q_j, j=1,\ldots,m$ (a) let the atomic formulae $Q_j(u^j_r), r=1,\ldots,r(j)$, be all the occurrences of atomic $Q_j$-formulae in $\mathrm{AT}$ and (b) let the formulae $\forall^1w^j_t(x_t(\mathbf{R}''_\Box)^{k_t+1}w^j_t\lra Q_j(w^j_t))$, $t=1,\ldots,t(j)$, be all the occurrences in $\mathrm{BOXED}$ of translations of boxed sentences $\bb^{k_t+1}Q_j$. A model $\mathfrak{M}=(\mathfrak{F},V, Q^\mathfrak{M}_j,\ldots)$ is a model with a minimal interpretation for the predicate variable $Q_j$ if $Q^\mathfrak{M}_j$
is the set $Q^\mathfrak{M}_j=\{V(u^j_1),\ldots,V(u^j_{r(j)})\}\cup\bigcup_{t=1}^{t(j)}V(x_t)(R''_\Box)^{k_t+1}$, where $V(x_t)(R''_\Box)^{k_t+1}$ is the set of points $z\in Z_1$ such that $V(x_t)(R''_\Box)^{k_t+1}z$ obtains in the frame. In other words, for $w\in Z_1$, we choose to interpret $Q_j$ so that $w\in Q_j^\mathfrak{M}$ iff $w=V(u^j_r)$, for some $1\leq r\leq r(j)$, or $V(x_t)(R''_\Box)^{k_t+1}w$ holds, for some $1\leq t\leq t(j)$.

To route the syntactic re-write rules of the Sahlqvist - Van Benthem algorithm, define  
\begin{equation}\label{lambda-Q}
\lambda(Q_j)=\lambda s.\left(\bigvee_{r=1}^{r(j)}(s=u^j_r)\right)\vee\left(\bigvee_{t=1}^{t(j)}x_t(\mathbf{R}''_\Box)^{k_t+1}s\right)
\end{equation}
so that in the model $\mathfrak{M}=(\mathfrak{F},V, Q^\mathfrak{M}_j,\ldots)$ we shall have that $\mathfrak{F}\models\lambda(Q_j)u[V]$ iff $V(u)\in Q_j^\mathfrak{M}$. Phrased differently, in the model $\mathfrak{M}=(\mathfrak{F},V, Q^\mathfrak{M}_j,\ldots)$ the formula $\lambda(Q_j)u$ evaluates to $\mathrm{TRUE}$ iff $V(u)\in Q_j^\mathfrak{M}$ (i.e. $\lambda(Q_j)$ designates the characteristic function of $Q^\mathfrak{M}_j$).

\subparagraph{Constrained Variables:}  For each $P_i$, for $i=1,\ldots, n+k$ (for which a corresponding $\mathrm{t}${-}invariance constraint is in $\mathrm{t{-}INV}$) (a) let $P_i(z^i_p), p=1,\ldots, p(i)$  be all the occurrences of atomic $P_i$-formulae in $\mathrm{AT}$ and (b) let the formulae $\forall^1w^i_q(x_q(\mathbf{R}''_\Box)^{k_q+1}w^i_q\lra P_i(w^i_q))$, $q=1,\ldots,q(i)$, be all the occurrences in $\mathrm{BOXED}$ of translations of boxed sentences $\bb^{k_q+1}P_i$. 

Note that given a model $\mathfrak{M}=(\mathfrak{F},V)$ and subset $U\subseteq Z_1$, the point $V(u)\in Z_1$ is in the Galois closure $\lbboxi\ldi U=U''$ of $U$ iff $\forall y\in Z_\partial[V(u)Iy\lra\exists z\in Z_1(zIy\wedge z\in U)]$. 
In particular, consider the set 
\[
U_i=\val{\bigvee_{p=1}^{p(i)}(z=z^i_p)\vee\bigvee_{q=1}^{q(i)}x_q(\mathbf{R}''_\Box)^{k_q+1}z}_\mathfrak{M}= \bigcup_{p=1}^{p(i)}\{V(z^i_p)\}\cup\bigcup_{q=1}^{q(j)}V(x_q)(R''_\Box)^{k_q+1}. 
\]
We choose to interpret $P_i$ as the least Galois stable set containing $U_i$ (the Galois closure of $U_i$). In accordance to this choice, we specify the characteristic function $\lambda(P_i)$ by setting
\begin{equation}\label{lambda-P}
\lambda(P_i)=\lambda s.\forall^\partial y\left(s\mathbf{I}y\lra\exists^1 z \left(z\mathbf{I}y \wedge\left[\bigvee_{p=1}^{p(i)}(z=z^i_p)\vee\bigvee_{q=1}^{q(i)}x_q(\mathbf{R}''_\Box)^{k_q+1}z\right]\right)\right).
\end{equation} 
Hence, $\mathfrak{F}\models\lambda(P_i)u[V]$ (i.e.  $\lambda(P_i)u$ evaluates to $\mathrm{TRUE}$ in the model $\mathfrak{M}$),  iff $V(u)$ is in the Galois closure of the set $U_i=\val{\bigvee_{p=1}^{p(i)}(z=z^i_p)\vee\bigvee_{q=1}^{q(i)}x_q(\mathbf{R}''_\Box)^{k_q+1}z}_\mathfrak{M}$.

\paragraph{Eliminating Second-Order Quantifiers.}
Replace occurrences of $P_i, Q_j$ in the formula~\eqref{step1-modified so} by the designations  $\lambda(P_i)$ and $\lambda(Q_j)$ of the characteristic functions of their interpretations and perform $\beta$-reduction. 

The above action results in eliminating the occurrences of the predicate variables $P_i,Q_j$ in $\mathrm{t{-}INV, AT, BOXED}$ and each of them evaluates to $\mathrm{TRUE}$.  We provide some details.
\begin{itemize}
  \item[($\mathrm{t{-}INV}$):]\hskip1cm Let $P$ be one of the predicate variables $P_i$ with an invariance constraint in  ($\mathrm{t{-}INV}$). Substituting $\lambda(P)$ for $P$ in equation~\eqref{t-invariance condition 1} we obtain
      \[
\forall^1 u[\forall^\partial v(u\mathbf{I}v\lra\exists^1 u_1(u_1\mathbf{I}v\wedge {\lambda(P)}u_1))\lra{\lambda(P)}u]
      \]
which evaluates to $\mathrm{TRUE}$ just in case the closure of the interpretation of $P$ is contained in the interpretation of $P$.
By $\beta$-conversion we obtain
\[
\lambda(P)u=\forall^\partial y\left(u\mathbf{I}y\lra\exists^1 z \left(z\mathbf{I}y \wedge\left[\bigvee_{p=1}^{p(i)}(z=z^i_p)\vee\bigvee_{q=1}^{q(i)}x_q(\mathbf{R}''_\Box)^{k_q+1}z\right]\right)\right),
\] 
i.e. $\lambda(P)u$ interprets to $\mathrm{TRUE}$ in a model $\mathfrak{M}=(\mathfrak{F},V)$ 
iff $V(u)$ is in the closure of the interpretation of the set $U=\val{\bigvee_{p=1}^{p(i)}(z=z^i_p)\vee\bigvee_{q=1}^{q(i)}x_q(\mathbf{R}''_\Box)^{k_q+1}z}_\mathfrak{M}$. By choice of the minimal interpretation, $P$ interprets to a stable set (the closure of the set $U$) and, thereby, the corresponding $\mathrm{t}${-}invariance constraint for $P$ interprets to $\mathrm{TRUE}$.

  \item [($\mathrm{AT}$):] \hskip1.1cm Let $P_i, Q_j$ be predicate variables, where $\mathrm{t{-}INV}$ contains a $\mathrm{t}${-}invariance constraint $\forall^1u_i[\mathrm{t}(P_i)(u_i)\ra P_i(u_i)]$  and let $P_i(z^i_{pp}), Q_j(u^j_{rr})$ be atomic formulae in $\mathrm{AT}$, for some $1\leq pp\leq p(i)$ and $1\leq rr\leq r(j)$. Given the definition of the characteristic functions $\lambda(Q_j)$ in~\eqref{lambda-Q} and $\lambda(P_i)$ in~\eqref{lambda-P} and replacing $P_i(z^i_{pp}), Q_j(u^j_{rr})$ in $\mathrm{AT}$ by $\lambda(P_i)z^i_{pp}$ and $\lambda(Q_j)u^j_{rr}$, respectively, we obtain after $\beta$-reduction
      \begin{align}
        \lambda(Q_j)u^j_{rr} & = \; \bigvee_{r=1}^{r(j)}(u^j_{rr}=u^j_r)\vee\left(\bigvee_{t=1}^{t(j)}x_t(\mathbf{R}''_\Box)^{k_t+1}u^j_{rr}\right),\label{lambda-Q-beta}\\
        \lambda(P_i)z^i_{pp} & = \; \forall^\partial y\left(z^i_{pp}\mathbf{I}y\lra\exists^1 z \left(z\mathbf{I}y \wedge\left[\bigvee_{p=1}^{p(i)}(z=z^i_p)\vee\bigvee_{q=1}^{q(i)}x_q(\mathbf{R}''_\Box)^{k_q+1}z\right]\right)\right). \label{lambda-P-beta}
      \end{align}
      Then, for some index $1\leq rr=r\leq r(j)$, $\lambda(Q_j)u^j_{rr}\equiv (u^j_r=u^j_r)\equiv\mathrm{TRUE}$. For the constrained $P_i$, we may re-write the right-hand side of~\eqref{lambda-P-beta} as follows
      \begin{equation}\label{lambda-P-beta2}
  \forall^\partial y\left(z^i_{pp}\mathbf{I}y\lra\exists^1 z \left(z\mathbf{I}y \wedge\left[(z=z^i_{pp})\vee\bigvee_{p=1}^{p(i),p\neq pp}(z=z^i_p)\vee\bigvee_{q=1}^{q(i)}x_q(\mathbf{R}''_\Box)^{k_q+1}z\right]\right)\right).     
      \end{equation}
      Now $\lambda(P_i)z^i_{pp}$ evaluates to $\mathrm{TRUE}$ in a model $\mathfrak{M}=(\mathfrak{F},V)$ just in case the point $V(z^i_{pp})$ assigned to the individual variable $z^i_{pp}$ is in the set ${P_i}^\mathfrak{M}$ interpreting the predicate variable $P_i$. But observing that, by~\eqref{lambda-P-beta2}, $P_i$ interprets as the closure of a set containing the point $V(z^i_{pp})$, we may conclude that $\lambda(P_i)z^i_{pp}$ evaluates to $\mathrm{TRUE}$ in any model $\mathfrak{M}=(\mathfrak{F},V, P_i^\mathfrak{M},\ldots)$ with the minimal interpretation of $P_i$.
      
  \item[($\mathrm{BOXED}$):] \hskip 1.1cm Let $P_i, Q_j$ be predicate variables, where $\mathrm{t{-}INV}$ contains a $\mathrm{t}${-}invariance constraint $\forall^1u_i[\mathrm{t}(P_i)(u_i)\ra P_i(u_i)]$  and let  $\forall^1w^i_{qq}(x_{qq}(\mathbf{R}''_\Box)^{k_{qq}+1}w^i_{qq}\lra P_i(w^i_{qq}))$, for some $qq\in\{1,\ldots,q(i)\}$, and  $\forall^1w^j_{tt}(x_{tt}(\mathbf{R}''_\Box)^{k_{tt}+1}w^j_{tt}\lra Q_j(w^j_{tt}))$, for some $tt\in\{1,\ldots,t(j)\}$ be the translations of boxed atoms in $\mathrm{BOXED}$. Replace $P_i,Q_j$ by $\lambda(P_i)$ and $\lambda(Q_j)$, respectively.
      
      For $Q_j$, the corresponding formula becomes
      \[
      \forall^1w^j_{tt}\left(x_{tt}(\mathbf{R}''_\Box)^{k_{tt}+1}w^j_{tt}\lra \left(\bigvee_{r=1}^{r(j)}(s=u^j_r)\right)\vee\left(\bigvee_{t=1}^{t(j)}x_t(\mathbf{R}''_\Box)^{k_{t}+1}w^j_{t}\right)\right)
      \] 
      and then the antecedent $x_{tt}(\mathbf{R}''_\Box)^{k_{tt}+1}w^j_{tt}$ is a disjunct of the consequent, hence the formula evaluates to $\mathrm{TRUE}$.
      
      For $P_i$, the corresponding formula becomes
      \[
      \forall^1w^i_{qq}\left(x_{qq}(\mathbf{R}''_\Box)^{k_{qq}+1}w^i_{qq}{\ra}  \forall^\partial y\left(w^ i_{qq}\mathbf{I}y{\ra}\exists^1 z \left(z\mathbf{I}y \wedge\left[\bigvee_{p=1}^{p(i)}(z=z^i_p)\vee\bigvee_{q=1}^{q(i)}x_q(\mathbf{R}''_\Box)^{k_q+1}z\right]\right)\right)   \right).
      \]
      For any model $\mathfrak{M}=(\mathfrak{F},V)$, the above formula evaluates to $\mathrm{TRUE}$ if $V(x_{qq})(R''_\Box)^{k_{qq}+1}\subseteq (\bigcup_{p=1}^{p(i)}\{V(z^i_p)\}\cup\bigcup_{q=1}^{q(i)}V(x_q)(R''_\Box)^{k_q+1})''=\lbboxi\ldi    (\bigcup_{p=1}^{p(i)}\{V(z^i_p)\}\cup\bigcup_{q=1}^{q(i)}V(x_q)(R''_\Box)^{k_q+1})$. The inclusion indeed obtains since for the value $qq$ of the index $q$ the set $V(x_{qq})(R''_\Box)^{k_{qq}+1}$ is one of the sets whose union is taken on the right-hand side, hence it is contained in the closure of this union.
      
\end{itemize}

\paragraph{Correctness Proof.}
Given a Sahlqvist inequality (sequent) 
\begin{equation}\label{inequality}
  \alpha_1\leq_1\eta_1
\end{equation}
the reduction steps of Lemma~\ref{R1-R6} produce a system in canonical Sahlqvist form
\begin{equation}\label{reduced}
  \lset \mathrm{STB,CVC}\midsp\alpha \leq_1\eta \rset
\end{equation}
with guarded second order translation
\begin{equation}\label{guarded}
\forall^1 P_1\cdots\forall^1 P_n\forall^1 P_{n+1}\cdots\forall^1 P_{n+k}\forall^1 Q^*_1\cdots\forall^1 Q^*_m\forall^1 x(\mathrm{t{-}INV}\wedge\;\mathrm{ST}_x(\alpha )\lra\mathrm{ST}_x(\eta )).
\end{equation} 
Pulling out existential quantifiers, \eqref{guarded} is transformed to \eqref{step1} below
\begin{equation}\label{step1}
\forall^1 P_1\cdots\forall^1P_{n+k}\forall^1Q^*_1\cdots\forall^1Q^*_m\forall^1 x\forall^1z_1\cdots\forall^1 z_r(\mathrm{t{-}INV}\wedge\mathrm{REL}\wedge\mathrm{AT}\wedge\mathrm{BOXED}\lra\mathrm{POS}).
\end{equation}
Having defined minimal instantiations, substituting $\lambda(P_i)$ for $P_i$ and $\lambda(Q_j)$ for $Q_j$, performing a subsequent $\beta$-reduction,
 removing the thereby redundant second-order quantifiers and letting $\lambda(\mathrm{POS})$ designate the result in the consequent position, the formula below
\begin{equation}\label{final corr}
\forall^1z_1\cdots\forall^1 z_r(\mathrm{REL}\lra\lambda(\mathrm{POS}))
\end{equation}
is returned as the local first-order correspondent.

Proving correctness amounts to showing that \eqref{inequality}--\eqref{final corr} are all semantically equivalent. 

Equivalence of \eqref{inequality} and \eqref{reduced} was shown in Lemma~\eqref{R1-R6}.

\eqref{reduced} and~\eqref{guarded} are clearly equivalent, given the semantic understanding that $\alpha \leq_1\eta $ in \eqref{reduced} will not be validated in a model unless all constrained variables in $\mathrm{STB,CVC}$ are interpreted as Galois stable sets.  In the guarded second-order translation \eqref{guarded} this requirement is equivalently enforced by the $\mathrm{t}${-}invariance constraints.

Formulae \eqref{guarded} and \eqref{step1} are certainly equivalent since one is the transformation of the other by using established prenex form equivalences. 

The computed first-order local correspondent \eqref{final corr} is true in a model $\mathfrak{M}=(\mathfrak{F},V)$ in which \eqref{step1} is true, since the former is an instance of the latter. Conversely, assume \eqref{final corr} is true in $\mathfrak{M}$ and that the antecedent $\mathrm{t{-}INV}\wedge\mathrm{REL}\wedge\mathrm{AT}\wedge\mathrm{BOXED}$ of \eqref{step1} is also true in $\mathfrak{M}$. Note first that since in particular $\mathfrak{M}\models\mathrm{t{-}INV}$, the predicate variables $P_i,i=1,\ldots,n+k$, are interpreted as Galois stable sets in $\mathfrak{M}$. Also, since $\mathfrak{M}\models\mathrm{REL}$, then by the assumption that \eqref{final corr} is true in $\mathfrak{M}$ it follows that $\lambda(\mathrm{POS})$ is also true in $\mathfrak{M}$. But by definition of the Sahlqvist property, $\stx{x}{\eta }=\mathrm{POS}$ is positive, hence its interpretation is monotone in $P_i$ and $Q_j$. By choice of $\lambda(P_i),\lambda(Q_j)$ as minimal interpretations it follows that $\mathfrak{M}\models\stx{x}{\eta }$, as well.
\end{proof}

To clarify the detail of calculating correspondents using the generalized Sahlqvist - Van Benthem approach, we provide a few application examples. 

\begin{rem}\label{simple minimal}
In the general  case, we defined the minimal interpretation for a constrained predicate variable $P$ (where a constraint $P''\leq_1P$ or $P=_{\sharp_{Q'}}Q'$ is in the system)
by setting
\begin{equation}\label{lambda-P-again}
\lambda(P)=\lambda s.\forall^\partial y\left(s\mathbf{I}y\lra\exists^1 z \left(z\mathbf{I}y \wedge\left[\bigvee_{p=1}^{n}(z=u_p)\vee\bigvee_{q=1}^{m}x_q(\mathbf{R}''_\Box)^{k_q+1}z\right]\right)\right),
\end{equation} 
where the atomic formulae $P(u_1),\ldots,P(u_n)$ occur in the standard first-order translation of the antecedent and the consequent of the sequent under examination and there are also boxed atoms $\bb^{k_q+1}P$, translating to corresponding formulae $\forall^1u_j(x_q(\mathbf{R}''_\Box)^{k_q+1} u_j\lra P(u_j))$.

For use in application examples, we list here some frequently occurring simple cases of the definition of a minimal instantiation for constrained variables $P$.  
\begin{description}
  \item[($m=0,n=1$)] A single atomic formula $P(u)$ occurs and there are no boxed atoms. Then the intended minimal interpretation $P^\mathfrak{M}$ in a model $\mathfrak{M}=(\mathfrak{F},V)$ is the least Galois stable set containing the point $V(u)$, i.e. $P^\mathfrak{M}=\{V(u)\}''=\lbboxi\ldi\{V(u)\}=\{w\in Z_1\midsp V(u)\leq w\}$, so that we may simply define $\lambda(P)=\lambda s.u\leq s$, where by definition of the order $u\leq w$ iff $\forall y(wIy\lra uIy)$.
  \item[($m=0,n>1$)] $\lambda(P)$ now needs to designate the characteristic function of the set $\{w_1,\ldots,w_n\}''$, where we set $w_i=V(u_i)$. Since 
      \[
      \{w_1,\ldots,w_n\}''=\lbboxi\ldi\{w_1,\ldots,w_n\}=\{w\in Z_1\midsp\forall y(wIy\lra\exists z(zIy\wedge z\in\{w_1,\ldots,w_n\}))
      \]
      it follows that we may just define $\lambda(P)=\lambda s.\forall^\partial y(s\mathbf{I}y\lra\exists^1 z(z\mathbf{I}w_1\vee\cdots\vee z\mathbf{I}w_n))$.

  \item[($m=1,n=0$)] Here the definition can be the same as in the classical case, i.e. $\lambda(P)=\lambda s.x_1\mathbf{R}''_\Box s$. This is because $\mathbf{R}''_\Box$ interprets to the relation $R''_\Box$ and for any $x$ the set $xR''_\Box$ is already a Galois stable set, as it was defined by setting $xR''_\Box =(xR'_\Box)'$. 
\end{description}
\end{rem}

\begin{ex}\label{box p proves p example}
Consider the sequent $\boxminus P\proves P$, as well as the sequent $\bb p\proves p$ in the regular fragment of the sorted modal logic system (equivalently, in the language of distribution-free modal logic). 

Treating the first is not different from its treatment in the classical case. Indeed, we have $\stx{x}{\boxminus P}=\forall^1 z(x\mathbf{R}''_\Box z\lra\mathbf{P}(z))$, hence for the second-order translation of the sequent we obtain $\forall^1 P\forall^1 x[\underbrace{\forall^1 z(x\mathbf{R}''_\Box z\lra P(z))}_{\mathrm{BOXED}}\lra \underbrace{P(x)}_{\mathrm{POS}}]$.

The minimal valuation that makes the antecedent of the implication true is obtained by setting $\lambda({P})=\lambda s.x\mathbf{R}''_\Box s$. In other words we interpret ${P}$ as the set $xR''_\Box$. Hence $\lambda({P})z=x\mathbf{R}''_\Box z$ and $\lambda({P})x=x\mathbf{R}''_\Box x$. We conclude with the familiar first-order equivalent $\Phi(x)= x\mathbf{R}''_\Box x$.

The sequent $\bb p\proves p$ in the regular fragment (equivalently, in the distribution-free modal logic) is a notational variant of the sequent $\boxminus P''\proves P''$. This is a Sahlqvist sequent, reducing to $\lset P''\leq_1P\midsp\boxminus P\leq_1 P\rset$. Its guarded second-order translation is 
\[
\forall^1P\forall^1x[\underbrace{\forall^1u[\mathrm{t}(P)(u)\lra P(u)]}_{\mathrm{t{-}INV}}\wedge \underbrace{\forall^1 z(x\mathbf{R}''_\Box z\lra P(z))}_{\mathrm{BOXED}}\lra \underbrace{P(x)}_{\mathrm{POS}}]. 
\]
For the minimal instantiation, taking also into consideration Remark~\ref{simple minimal}, we set $\lambda(P)=\lambda s.x\mathbf{R}''_\Box s$. Since $xR''_\Box=(xR'_\Box)'$ is Galois stable,  $\mathrm{t{-}INV}$ evaluates to $\mathrm{TRUE}$ and the same holds for $\mathrm{BOXED}$. The consequent, after $\beta$-reduction, is the formula $x\mathbf{R}''_\Box x$ and this is precisely the local first-order correspondent of the sequent $\bb p\proves p$ (see also the discussion and computation in \cite[Section~5.2]{dfnmlA}, where it was shown that the double dual relation $R''_\Box$ is reflexive iff $R^{\partial\partial}_\Box$ is reflexive).
\end{ex}

\begin{ex}\label{p proves diamond p example}
Consider the DfML sequent $p\proves\dd p$. In Example~\ref{both} we saw that its translation $P''\proves(\diamondvert P'')''$  reduces to 
 $\lset P''\leq_1 P\midsp P\leq_1(\diamondvert P)''$, which is in canonical Sahlqvist form. 
Its co-translation is the $\partial$-sequent $(\diamondvert P'')'\vproves P'$ which reduces to $\lset Q=_\partial P'\midsp \boxvert Q\leq_\partial Q\rset$, which is also in canonical Sahlqvist form. 

\paragraph{{\bf Co-translation sequent}.} 
The system $\lset Q=_\partial P'\midsp \boxvert Q\leq_\partial Q\rset$ is essentially the same system as in Example~\ref{box p proves p example}, except that it corresponds to a $\partial$-sequent and $\boxvert$ is the $\mathcal{L}_\partial$ box operator, with $\boxvert Q$ being interpreted using the double dual $R''_\Diamond$ of the frame relation $R^{11}_\Diamond$ (see Table~\ref{sorted sat table}). It follows, given Example~\ref{box p proves p example}, that the resulting local first-order correspondent is the formula $\Psi(y)=y\mathbf{R}''_\Diamond y$.

\paragraph{{\bf Translation sequent}.}
Given $P''\proves(\diamondvert P'')''$, we may consider directly the system $\lset P''\leq_1 P\midsp P\leq_1 (\diamondvert P)''\rset$ and the guarded second-order translation
\begin{equation}\label{so}
\forall^1 P\forall^1 x[\forall^1 u[t(P)(u)\lra P(u)]\wedge P(x)\lra\mathrm{ST}_x((\diamondvert P)''), \mbox{ where }
\end{equation}
  \begin{equation}\label{consequent}
 \mathrm{ST}_x((\diamondvert P)'')= \forall^\partial v[\mathbf{I}(x,v)\lra\exists^1 z(\mathbf{I}(z,v)\wedge\exists^1 u(z\mathbf{R}^{11}_\Diamond u\wedge P(u)))].
  \end{equation}
 For the minimal instantiation we set $\lambda(P)=\lambda s.x\leq s$, so that $u\in P$ iff $x\leq u$, which is what we obtain by the $\beta$-conversion $\lambda(P)u=(x\leq u)$.

For the $\mathrm{t}$-invariance constraint, after substituting $\lambda(P)$ for $P$, we get
\[
\forall^1 u[\forall^\partial v(\mathbf{I}(u,v)\lra\exists^1 u_1(\mathbf{I}(u_1,v)\wedge {\lambda(P)}u_1))\lra{\lambda(P)}u]
\]
and after $\beta$-reduction we obtain
\[
\forall^1 u[\forall^\partial v(\mathbf{I}(u,v)\lra\exists^1 u_1(\mathbf{I}(u_1,v)\wedge x\less u_1))\lra x\less u]
\]
which evaluates to $\mathrm{TRUE}$ because $\Gamma u$ is stable. Since also $\lambda(P)x=(x\leq x)$ we obtain
\[
\forall^\partial v[\mathbf{I}(x,v)\lra\exists^1 z(\mathbf{I}(z,v)\wedge\exists^1 u(z\mathbf{R}^{11}_\Diamond u\wedge x\leq u))].
\]
Assuming also the frame axiom (F3) on the monotonicity properties of $R^{11}_\Diamond$  we get the equivalent formula
\[
\forall^\partial v[\mathbf{I}(x,v)\lra\exists^1 z(\mathbf{I}(z,v)\wedge z\mathbf{R}^{11}_\Diamond x)].
\]
That this is equivalent to $R''_\Diamond$ being reflexive can be seen by the following calculation.
\begin{tabbing}
\hskip8mm\= $\forall^\partial v[xIv\lra\exists^1 z(zIv\wedge zR^{11}_\Diamond x)]$\\
iff\> $\forall^\partial v[\neg\exists^1 z(zIv\wedge zR^{11}_\Diamond x)\lra x\upv v]$\\
iff\> $\forall^\partial v[\forall^1 z(zR^{11}_\Diamond x\lra z\upv v)\lra x\upv v]$\\
iff\> $\forall^\partial v[vR^{\partial 1}_\Diamond x\lra x\upv v]$\\
iff\> $\forall^\partial v[vR'_\Diamond x\lra x\upv v]$\\
iff\> $\forall^\partial v \;v\in (vR'_\Diamond)'$\\
iff\> $\forall^\partial v\; vR''_\Diamond v$.
\end{tabbing}
It has been argued in \cite[Section~5.2]{dfnmlA} that $R''_\Diamond$ is reflexive iff $R_\Diamond$ is.
\end{ex}

\begin{ex}\label{s4 diamond}
In Example~\ref{only one is sahlqvist} we considered the sequent ${\dd\dd} p\proves\dd p$ in the language of DfML and we verified that the co-translation sequent reduces to  $\lset Q=_\partial P'\midsp \boxvert Q\leq_\partial(\diamondvert(\boxvert Q)')'\rset$, which is in canonical Sahlqvist form, while its translation sequent does not so reduce. 

Working out the example manually, notice that using the reduction rule (R5.1) $(\diamondvert\beta')'\mapsto\boxvert\beta''$, for $\beta=\boxvert Q\in\mathcal{L}_\partial$, results in obtaining the equivalent system $\lset Q=_\partial P'\midsp \boxvert Q\leq_\partial \boxvert(\boxvert Q)''\rset$. Finally, applying (R5.3) we get the simple Sahlqvist system $\lset Q=_\partial P'\midsp \boxvert Q\leq_\partial{\boxvert\boxvert}Q\rset$.

The guarded second-order translation is given by
\[
\forall^\partial Q\forall^\partial y[\underbrace{\forall^\partial v[\mathrm{t}(Q)(v){\ra} Q(v)]}_{\mathrm{t{-}INV}}\wedge\underbrace{\mathrm{ST}_y(\boxvert Q)}_{\mathrm{BOXED}}\lra\underbrace{\mathrm{ST}_y({\boxvert\boxvert}Q)}_{\mathrm{POS}}]
\]
where $\mathrm{ST}_y(\boxvert Q)=\forall^\partial w(y\mathbf{R}''_\Diamond w\lra Q(w))$. By the proof of Theorem~\ref{Sahlqvist thm}   and Remark~\ref{simple sahlqvist}   we may set $\lambda(Q)=\lambda s.y\mathbf{R}''_\Diamond s$, obtaining $\lambda(Q)w=y\mathbf{R}''_\Diamond w$. Since $yR''_\Diamond$ is a Galois set, the $\mathrm{t}$-invarianvce constraint evaluates to $\mathrm{TRUE}$, and by the choice of minimal interpretation the same holds for the $\mathrm{BOXED}$ conjunct of the antecedent of the implication. The consequent, after replacing $Q$ with $\lambda(Q)$ and $\beta$-reduction is the formula
\[
\forall^\partial w_1(y\mathbf{R}''_\Diamond w_1\lra\forall^\partial w_2(w_1\mathbf{R}''_\Diamond w_2\lra y\mathbf{R}''_\Diamond w_2))
\]
which is precisely a transitivity constraint for the double dual relation $R''_\Diamond$. In \cite[Section~ 5.2]{dfnmlA} it was argued and verified that the double dual relation $R''_\Diamond$ is transitive iff $R_\Diamond$ is transitive (and similarly for $R''_\Box$ and $R_\Box$).
\end{ex}

\begin{ex}
\label{K1-K2}
Classically, the K-axiom $\Box(p\ra q)\proves\Box p\ra\Box q$ is equivalent to each of (K1) $\Diamond p\wedge\Box q\proves\Diamond(p\wedge q)$ and (K2) $\Box(p\vee q)\proves\Diamond p\vee\Box q$, used in Dunn's Positive Modal Logic (PML) to axiomatize the interaction of $\Box$ and $\Diamond$. In a distributive setting the two axioms impose that both operators be interpreted by one and the same accessibility relation. 

The situation is not as straightforward in a non-distributive setting.  In \cite{choice-free-dmitrieva-bezanishvili}, (K1) is assumed in the axiomatization of the logic, but not (K2), which is not discussed at all in the correspondence Section~4.5. Conradie and Palmigiano \cite[Example~3.15,Example~5.5]{conradie-palmigiano} observe that though both (K1) and (K2) are Sahlqvist if distribution is assumed, the ALBA algorithm fails on either of them in a non-distributive setting. It is interesting to see that both (K1) and (K2) can be handled in the reductionist correspondence approach that we are taking here, as we show below.

\paragraph{(K1) $\Diamond p\wedge\Box q\proves\Diamond(p\wedge q)$}
Consider (K1) in distribution-free modal logic,  in out notation  ${\dd}p_1\wedge{\bb}p_2\proves{\dd}(p_1\wedge p_2)$, translating to $(\diamondvert P''_1)''\cap (\diamondminus P'_2)'\proves(\diamondvert(P''_1\cap P''_2))''$. By (R5.1), we can rewrite $(\diamondminus P'_2)'$ as $\boxminus P''_2$. By (R4) applied for both variables we obtain the system 
\[ 
\{P''_1\leq_1P_1,P''_2\leq_1P_2\midsp (\diamondvert P_1)''\cap \boxminus P_2\leq_1(\diamondvert(P_1\cap P_2))''\}.
\]
Since $P''_2\leq_1P_2$ is an assumed constraint, rule (R9) applies and we obtain the equivalent system
\[
\{P''_1\leq_1P_1,P''_2\leq_1P_2\midsp {\diamondvert} P_1\cap \boxminus P_2\leq_1(\diamondvert(P_1\cap P_2))''\}
\]
which is in canonical Sahlqvist form. 

The guarded second-order translation, after pulling out the existential quantifier $\exists^1z$, is 
\[
\forall^1P_1\forall^1P_2\forall^1x\forall^1z[\mathrm{t{-}INV} \wedge x\mathbf{R}_\Diamond z\wedge P_1(z)\wedge\forall^1u(x\mathbf{R}''_\Box u\lra P_2(u))\lra \mathrm{POS} ],
\]
where $\mathrm{t{-}INV} = \bigwedge_{i=1}^2\forall^1u_i[\mathrm{t}(P_i)(u_i)\ra P_i(u_i)] $ and
\[
  \mathrm{POS} = \mathrm{ST}_x((\diamondvert(P_1\wedge P_2))'')=\forall^\partial y(x\mathbf{I}y\lra\exists^1w[w\mathbf{I}y\wedge \exists^1 u(z\mathbf{R}_\Diamond u\wedge P_1(u)\wedge P_2(u))]).
\]
In accordance to Theorem~\ref{Sahlqvist thm} and Remark~\ref{simple minimal}, we set $\lambda(P_1)=\lambda s.x\leq s$ and $\lambda(P_2)=\lambda s.x\mathbf{R}''_\Box s$ so that the invariance constraints evaluate to $\mathrm{TRUE}$ and we obtain the local correspondent
\[
\forall^1z\forall^\partial y[x\mathbf{R}_\Diamond z\wedge x\leq z\wedge x\mathbf{I}y\lra\exists^1w\exists^1 u(w\mathbf{I}y\wedge z\mathbf{R}_\Diamond u\wedge x\leq u\wedge x\mathbf{R}''_\Box u)].
\]

\paragraph{(K2) $\Box(p\vee q)\proves\Diamond p\vee\Box q$} 
The following is a reduction to canonical Sahlqvist form.
\begin{tabbing}
1.\hskip4mm\= $\langle (\dd p\vee\bb q)^\circ\leq_\partial (\bb(p\vee q))^\circ     \rangle$\hskip4.5cm\= \\
2.\> $\langle  (\dd p)^\circ\cap (\bb q)^\circ\leq_\partial (\diamondminus(p\vee q)^\circ)''    \rangle$\\
3.\> $\langle (\diamondvert P'')'\cap (\diamondminus Q')''\leq_\partial (\diamondminus(P'\cap Q')   )''      \rangle$ \> (R5.3)  $(\diamondvert\beta')'\mapsto\boxvert\beta''$, for $\beta=P'$\\
4.\> $\langle  \boxvert P'''\cap   (\diamondminus Q')''\leq_\partial (\diamondminus(P'\cap Q')   )''      \rangle$\> (R5.4)\\
5.\> $\langle   \boxvert P'\cap   (\diamondminus Q')''\leq_\partial (\diamondminus(P'\cap Q')   )''     \rangle$ \> (R6)\\
6.\> $\langle P_2=_\partial P',P_1=_\partial Q'\midsp  {\boxvert} P_2\cap   (\diamondminus P_1)''\leq_\partial (\diamondminus(P_2\cap P_1)   )''     \rangle$\> (R9)
\\
7.\> $\langle  P_2=_\partial P',P_1=_\partial Q'\midsp  {\boxvert} P_2\cap   \diamondminus P_1\leq_\partial (\diamondminus(P_1\cap P_2))''     \rangle$. 
\end{tabbing} 
Computing the correspondent is left to the interested reader.
\end{ex}

\subsection{Negation}
\label{negation section}
Recall from Section~\ref{dual algebras section} that the following identities hold, for any stable set $A$ and co-stable set $B$.
  \[
(\ltdown A)'=\lttdown A =\lbtdown A'\;\mbox{ and } \lttdown B'=\lbtdown B=(\ltdown B')'.
\]
In other words, the restriction of $\lbtdown:\powerset(Z_\partial)\lra\powerset(Z_1)$ to $\gphi$ (the lattice of co-stable sets) is the Galois dual operation of $\ltdown:\powerset(Z_1)\lra\powerset(Z_\partial)$.

This is what suggested introducing the re-write rule (R5.7) in Table~\ref{rewrite rules table}, displayed below for the reader's convenience.

\begin{quote}
\begin{enumerate}
\item[(R5.7)]\hskip2mm $(\tdown\alpha'')'$ $\mapsto$  $\boxtimes\alpha'$ and\\
\hspace*{1.2mm} $(\tdown P)'\mapsto\btdown P'$, provided $P''\leq_{\sharp_P}P$, or $P=_{\sharp_{Q'}}Q'$ is declared in $\mathrm{STB,CVC}$.
\end{enumerate}
\end{quote}

\begin{thm}\label{Sahlqvist thm with negation}
  Every Sahlqvist sequent in the implication-free fragment of DfML has a first-order local correspondent, effectively computable from the input sequent.
\end{thm}
\begin{proof}
The only amendment to the proof of the extended Sahlqvist -- Van Benthem theorem presented in Section~\ref{box and diamond section} that is further needed  is that we now have two box operators $\boxminus$ and $\btdown$ involved in boxed atoms, interpreted respectively by $\lbminus$, generated in frames by the relation $R''_\Box$, and $\lbtdown$, generated by the relation $R''_\tdown$. 
Having multiple box operators is however a familiar situation already in classical Sahlqvist theory. 

For each composite string of boxes $\Box^*{\in}\{\boxminus,\btdown\}^n$, of some length $n>0$, the standard translation $\mathrm{ST}_x(\Box^*P)$ involves a composition of the predicates $\mathbf{R}''_\Box$ and $\mathbf{R}''_\tdown$, for which we may introduce a predicate $\mathbf{R}''_*$, thus obtaining $\mathrm{ST}_x(\Box^*P)=\forall^1z(x\mathbf{R}''_*z\lra P(z))$. In the proof of Theorem~\ref{Sahlqvist thm}, in equation~\eqref{lambda-Q} for unconstrained variables and in equation~\eqref{lambda-P} for constrained variables, we just replace $\mathbf{R}''_\Box$ with $\mathbf{R}''_*$ and otherwise the same argument applies, which is exactly the strategy followed in the classical case. 
\end{proof}

\begin{ex}
Consider the Galois connection axiom in distribution-free modal logic, $p\proves\ttdown\ttdown p$. Recall first that $(\ttdown \varphi)^\bullet=(\tdown \varphi^\bullet)'$, so that $(\ttdown\ttdown p)^\bullet=(\tdown(\ttdown p)^\bullet)'=(\tdown(\tdown P'')')'$. Hence we have the reductions
\begin{tabbing}
1.\hskip4mm\=$\lset P''\leq_1(\tdown(\tdown P'')')'\rset$\\
2.\>$\lset P''\leq_1P\midsp P\leq_1(\tdown(\tdown P)')'\rset$.
\end{tabbing}
This translates to
\[
\forall^1P\forall^1x[\forall^1u[\mathrm{t}(P)(u)\lra P(u)]\wedge P(x)\lra\underbrace{\mathrm{ST}_{x}{((\tdown(\tdown P)')')}}_{\mathrm{POS}} ]
\]
where $\mathrm{POS}=\forall^\partial y(x\mathbf{I}y\lra\forall^1z(yR^{\partial 1}_\tdown z\lra\exists^\partial v(z\mathbf{I}v\wedge\exists^1u(v\mathbf{R}^{\partial 1}_\tdown u\wedge P(u)))))$. 

We let $\lambda(P)=\lambda s.x\leq s$. After $\beta$-reduction the antecedent evaluates to $\mathrm{TRUE}$ and we obtain the first-order local correspondent
\[
\forall^\partial y(x\mathbf{I}y\lra\forall^1z(yR^{\partial 1}_\tdown z\lra\exists^\partial v(z\mathbf{I}v\wedge\exists^1u(v\mathbf{R}^{\partial 1}_\tdown u\wedge x\leq u))))
\]
By the monotonicity properties of frame relations (assuming also the frame axiom (F3)) we have $\exists^1u(v\mathbf{R}^{\partial 1}_\tdown u\wedge x\leq u)$ iff $v\mathbf{R}^{\partial 1}x$. Hence we obtain the equivalent formula
\begin{equation}\label{galois negation}
\forall^\partial y(x\mathbf{I}y\lra\forall^1z(yR^{\partial 1}_\tdown z\lra\exists^\partial v(z\mathbf{I}v\wedge v\mathbf{R}^{\partial 1}x)))
\end{equation}
\end{ex}

In \cite{choiceFreeStLog}, where a choice-free duality for lattices with a weak complementation operator $\nu$ was presented, we also proved related correspondence results, in \cite[Corollary~3.15]{choiceFreeStLog}. It was shown, in particular, that the axiom imposing that $\nu$ forms a Galois connection with itself, i.e. $a\leq\nu\nu a$ for all lattice elements $a$, defines the class of sorted residuated frames in which the Galois dual $R'_\tdown$ of the frame relation $R^{\partial 1}$ is symmetric. For the reader's benefit we show that symmetry of $R'_\tdown$ is equivalent to the local correspondent~\eqref{galois negation} that we computed above.

We first show that \eqref{galois negation} is equivalent to the claim that the composite relation $I\circ R^{\partial 1}_\tdown$ is symmetric.
\begin{tabbing}
  \hskip1cm \=  
    $\forall^\partial y(x\mathbf{I}y\lra\forall^1z(yR^{\partial 1}_\tdown z\lra\exists^\partial v(z\mathbf{I}v\wedge v\mathbf{R}^{\partial 1}x)))$\\
  iff \> $\forall^\partial y\forall^1z(x\mathbf{I}y\lra(yR^{\partial 1}_\tdown z\lra\exists^\partial v(z\mathbf{I}v\wedge v\mathbf{R}^{\partial 1}x)))$\\
  iff \> $\forall^\partial y\forall^1z(x\mathbf{I}y\wedge yR^{\partial 1}_\tdown z\lra\exists^\partial v(z\mathbf{I}v\wedge v\mathbf{R}^{\partial 1}x)))$\\
  iff \> $\forall^1z[\exists^\partial y(x\mathbf{I}y\wedge yR^{\partial 1}_\tdown z)\lra\exists^\partial v(z\mathbf{I}v\wedge v\mathbf{R}^{\partial 1}x))]$\\
  iff \> $\forall^1z[x(I\circ R^{\partial 1}_\tdown) z\lra z(I\circ R^{\partial 1}_\tdown)x]$
\end{tabbing}
We next verify that $I\circ R^{\partial 1}_\tdown$ is symmetric iff the Galois dual relation $R'_\tdown$ is symmetric. Sketching the argument, assume first that $I\circ R^{\partial 1}$ is symmetric, assume also that $xR'_\tdown z$ holds, but suppose, for a contradiction, that $zR'_\tdown x$ fails. But $\neg(zR'_\tdown x)$ is equivalent to $\neg\forall^\partial y(yR^{\partial 1}_\tdown x\lra z\upv y)$, in turn equivalent to $\exists^\partial y(yR^{\partial 1}_\tdown x\wedge zIy)$, which is the same as $z(I\circ R^{\partial 1}_\tdown)x$. Symmetry for $I\circ R^{\partial 1}$ is assumed, hence  $x(I\circ R^{\partial 1}_\tdown)z$ also obtains. Unfolding conditions backwards, this is equivalent to $\exists^\partial y(xIy\wedge yR^{\partial 1}_\tdown z)$, in turn equivalent to $\neg\forall^\partial y(yR^{\partial 1}_\tdown z\lra x\upv y)$. But this is the same as negating the assumption that $xR'_\tdown z$ holds. We leave to the interested reader the argument for the converse direction, i.e. that symmetry of the Galois dual relation $R'_\tdown$ implies symmetry of the composite relation $I\circ R^{\partial 1}_\tdown$.

\begin{ex}
  Consider the axiom $\ttdown\ttdown p\proves p$, translating to $(\tdown(\tdown P'')')'\proves P''$ and co-translating to $P'''\vproves (\tdown(\tdown P'')')''$. The first cannot be reduced to canonical Sahlqvist form. For the $\partial$-sequent (the co-translation), we have the following reduction:
  \begin{tabbing}
1.\hskip4mm\=  $\lset P'\leq_\partial (\tdown(\tdown P'')')''\rset$\hskip3cm\=
  \\
2.\>  $\lset P'\leq_\partial (\tdown \boxtimes P''')''\rset$ \> (R5.7) with $\alpha=P$
  \\
3.\>  $\lset P'\leq_\partial (\tdown\boxtimes P')''\rset$ \> (R5.4)
  \\
4.\>   $\lset Q=_\partial P'\midsp Q\leq_\partial(\tdown\btdown Q)''\rset$ \> (R6)
  \end{tabbing}
Computing the correspondent from the canonical Sahlqvist form is left to the reader.
\end{ex}

We end this section by mentioning two non-examples.
\begin{ex}
The Kleene negation axiom $p\wedge\ttdown p\proves q\vee\ttdown q$ is not Sahlqvist. It translates to $P''\cap (\tdown P'')'\proves (Q''\cup (\tdown Q'')')''$ and co-translates to $(Q''\cup (\tdown Q'')')'''\vproves(P''\cap (\tdown P'')')'$. The 1-sequent reduces to $\lset P''\leq_1P, Q''\leq_1Q\midsp P\cap(\tdown P)'\leq_1(Q\cup(\tdown Q)')''\rset$. Using the re-write rule $(\tdown P)'\mapsto\btdown P'$ does not remove the difficulty since variables occur both primed and unprimed. Essentially the same difficulty appears in processing the co-translation.

The intuitionistic principle $p\wedge\ttdown p\proves\bot$,
 translating to $P''\cap (\tdown P'')'\leq_1\bot$, reduces to 
$\lset P''\leq_1 P\midsp P\cap(\tdown P)'\leq_1\bot\rset$, which cannot be further reduced to canonical Sahlqvist form because $P$ occurs both positively and negatively.
  
The co-translation $\top\vproves (P''\cap (\tdown P'')')'$, reduces to
$\lset P''\leq_1 P\midsp \top\vproves (P\cap(\tdown P)')'\rset$,
 or after re-writing $\lset P''\leq_1 P\midsp (P\cap\btdown P')'\rset$, again not in canonical Sahlqvist form, since the right-hand-side of the main inequality is not positive.
\end{ex}

\subsection{Implication}
\label{implication section}
Generalizing the correspondence algorithm to the case of the full language of DfML, thus including implication as well, is immediate. This is because our reduction strategy eliminates implication altogether, in favour of the additive (diamond) operators $\odot$, or $\tright$. Reversing our so far presentation choices, we start with examples.

Recall first from Section~\ref{dual algebras section} that the frame relation $T^{\partial 1\partial}$ generates a sorted normal additive operator $\ltright$, defined by
\begin{equation}\label{arrow image op}
U{\largetriangleright} V=\{y\in Y\midsp\exists x,v(x\in U\;\wedge\;v\in V\;\wedge\;yTxv)\}=\bigcup_{x\in U}^{v\in V}Txv
\end{equation}
and a stable sets implication operator was defined by $A\Ra C=(A\ltright C')'$, for $A,C\in\gpsi$.

Recall also that the double dual relation (Definition~\ref{double dual relations}) $R^{111}$ of $T^{\partial 1\partial}$ generates a binary normal additive operator $\bigodot$ on $\powerset(Z_1)$, residuated with an implicative construct $\Ra_T$ on $\powerset(Z_1)$ which coincides with $\Ra$ when restricted to stable sets (Proposition~\ref{Ra long and short}).

Because of the above we have two equivalent ways to translate an implicative sentence in the language of DfML, namely $(p\rfspoon q)^\bullet=(p^\bullet\tright q^\circ)'=(P''\tright Q')'=P''\rspoon Q''$.

\begin{ex}
Consider the contraction sequent $p\rfspoon(p\rfspoon q)\proves p\rfspoon q$ in the language of distribution-free modal logic. 

\paragraph{Choose $(p\rfspoon q)^\bullet=p^\bullet\rspoon q^\bullet$}
The following gives a reduction sequence to a system in canonical Sahlqvist form.
  \begin{tabbing}
1.  \hskip4mm\= $\langle P''\rspoon(P''\rspoon Q'')\leq_1 P''\rspoon Q''\rangle$ \hskip3cm\=\\
2.  \> $\langle P''\leq_1 P,Q''\leq_1 Q\midsp P\rspoon(P\rspoon Q)\leq_1 P\rspoon Q\rangle$\> (R4)\\
3.  \> $\langle P''\leq_1 P,Q''\leq_1 Q\midsp P\odot P\rspoon Q\leq_1 P\rspoon Q\rangle$ \>  (R5.8)\\
4.  \> $\langle P''\leq_1 P,Q''\leq_1 Q\midsp P\leq_1 P\odot P\rangle$ \>  (R8)\\
5.  \> $\langle P''\leq_1 P\midsp P\leq_1 P\odot P\rangle$ \>  (R1)
  \end{tabbing}
The interested reader can verify that the generalized Sahlqvist -- Van Benthem algorithm returns the formula $\exists^1u\exists^1z(xR^{111}uz\wedge x\leq u\wedge x\leq z)$ as a local first-order correspondent for the contraction sequent.

\paragraph{Choose $(p\rfspoon q)^\bullet=(p^\bullet\tright q^\circ)'$}
  The translation of the given sequent is $(P''\tright(P''\tright Q''')')'\proves (P''\tright Q''')'$. Introducing stability constraints we have a reduction to $\lset P''\leq_1P,Q''\leq_1Q\midsp (P\tright(P\tright Q')')'\leq_1 (P\tright Q')'\rset$. Introducing a change-of-variables constraint and discarding the thereby redundant related stability constraint we obtain the system of inequalities $\lset P''\leq_1 P,Q_1=_\partial Q'\midsp (P \tright (P\tright Q_1)')'\leq_1 (P\tright Q_1)'\rset$, hence the computation thread that picks to process the translation 1-sequent fails to reduce it to a system in canonical Sahlqvist form.
  
  For the co-translation (dual translation), recall first that $(\varphi\rfspoon\psi)^\circ$ = $(\varphi^\bullet\tright\psi^\circ)''$ (cf Table~\ref{syntactic translation into sorted}). Hence $(p\rfspoon(p\rfspoon q))^\circ=(p^\bullet\tright(p\rfspoon q)^\circ)''=(p^\bullet\tright(p^\bullet\tright q^\circ)'')''=(P''\tright(P''\tright Q')'')''$  and the following reduction steps 
  \begin{tabbing}
1.  \hskip4mm\= $\lset (P''\tright Q')''\leq_\partial(P''\tright(P''\tright Q')'')''\rset$\\
2.  \> $\lset P''\leq_1P\midsp (P\tright Q')''\leq_\partial(P\tright(P\tright Q')'')''\rset$\\
3.  \> $\lset P''\leq_1P,Q_1=_\partial Q'\midsp (P\tright Q_1)''\leq_\partial(P\tright(P\tright Q_1)'')''\rset$\\
4.  \> $\lset P''\leq_1P,Q_1=_\partial Q'\midsp P\tright Q_1\leq_\partial(P\tright(P\tright Q_1)'')''\rset$
  \end{tabbing}
succeed in reducing it to canonical Sahlqvist form. The guarded second-order translation is  
 \[
 \forall^1P\forall^\partial Q_1\forall^\partial y\forall^\partial v\forall^1 z[\mathrm{t{-}INV}\wedge y\mathbf{T}^{\partial 1\partial}zv\wedge P(z)\wedge Q_1(v)\lra\mathrm{ST}_y((P{\tright}(P{\tright}Q_1)'')'')]
 \]
 and we trust the reader to compute the resulting local correspondent.
\end{ex}

\begin{ex}
Consider the weakening axiom $p\proves q\rfspoon p$. None of the translation $P''\proves (Q''\tright P')'$ or co-translation $(Q''\tright P')''\vproves P'$ is Sahlqvist.

In the alternative (equivalent) translation we obtain the inequality $P''\leq_1 Q''\rspoon P''$. This is further equivalent to $Q''\odot P''\leq_1 P''$, reducing to $\lset P''\leq_1P,Q''\leq_1Q\midsp Q\odot P\leq_1 P\rset$, which is in canonical Sahlqvist form.
\end{ex}

\begin{ex}
  Consider the exchange axiom $p_1\rfspoon(p_2\rfspoon p_3)\proves p_2\rfspoon(p_1\rfspoon p_3)$. The following is a reduction sequence to canonical Sahlqvist form, choosing the translation $(p\rfspoon q)^\bullet=P''\rspoon Q''$.
\begin{tabbing}
1.\hskip4mm\= $\langle P_1''\rspoon(P_2''\rspoon P_3'')\leq_1 P_2''\rspoon(P_1''\rspoon P_3'') \rangle$ \hskip3cm\= \\
2.\> $\langle P_i''\leq_1P_i (i=1,2,3)\midsp P_1\rspoon(P_2\rspoon P_3)\leq_1 P_2\rspoon(P_1\rspoon P_3) \rangle$ \> (R4)\\
3.\> $\langle P_i''\leq_1P_i (i=1,2,3)\midsp P_2\odot P_1\rspoon P_3\leq_1 P_1\odot P_2\rspoon P_3\rangle$\> (R5.9), Table~\ref{rewrite rules table}\\
4.\> $\langle P_1''\leq_1P_1,P_2''\leq_1P_2\midsp P_1\odot P_2\leq_1P_2\odot P_1\rangle$. \> (R8)
\end{tabbing}
\end{ex}

\begin{ex}
As a last example, we consider the Fisher-Servi \cite{servi-Stud} axioms for IML, $\dd(p\rfspoon q)\proves\bb p\rfspoon\dd q$ and $\dd p\rfspoon\bb q\proves\bb(p\rfspoon q)$. The second axiom demonstrates the usefulness of having residuals $\bbox$ for the diamond operators in the language of the companion sorted modal logic (the reduction language).

For the first, we have a reduction
\begin{tabbing}
1.\hskip4mm\= $\langle (\bb p\rfspoon\dd q)^\circ\leq_\partial (\dd(p\rfspoon q))^\circ \rangle$\\
2.\> $\langle (\bb p)^\bullet\tright (\dd q)^\circ)''\leq_\partial \boxvert (p\rfspoon q)^\circ\rangle$\\
3.\> $\langle (\boxminus P''\tright {\boxvert} Q')''\leq_\partial  \boxvert(P''\tright Q')''   \rangle$\\
4.\> $\langle P''\leq_1 P,Q_1=_\partial Q'\midsp (\boxminus P\tright {\boxvert} Q_1)''\leq_\partial  \boxvert(P\tright Q_1)'' \rangle$\\
5.\> $\langle P''\leq_1 P,Q_1=_\partial Q'\midsp {\boxminus} P\tright {\boxvert}Q_1\leq_\partial  \boxvert(P\tright Q_1)'' \rangle$
\end{tabbing}
which is in canonical Sahlqvist form. 

Next consider the following reduction for $\dd p\rfspoon\bb q\proves\bb(p\rfspoon q)$.
\begin{tabbing}
1.\hskip4mm\= $\langle (\bb(p\rfspoon q))^\circ\leq_\partial (\dd p\rfspoon\bb q)^\circ \rangle$\\
2.\> $\langle (\diamondminus(p\rfspoon q)^\circ)''\leq_\partial ( (\dd p)^\bullet\tright(\bb q)^\circ    )'' \rangle$\\
3.\> $\langle (\diamondminus(P''\tright Q')''   )''\leq_\partial  ( (\diamondvert P'')''\tright (\diamondminus Q')''  )'' \rangle$\\
4.\> $\langle P''\leq_1P,Q_1=_\partial Q'\midsp (\diamondminus(P\tright Q_1)''   )''\leq_\partial  ( (\diamondvert P)''\tright (\diamondminus Q_1)''  )'' \rangle$\\
5.\> $\langle P''\leq_1P,Q_1=_\partial Q'\midsp \diamondminus(P\tright Q_1)'' \leq_\partial  ( (\diamondvert P)''\tright (\diamondminus Q_1)''  )''  \rangle$\\
6.\> $\langle P''\leq_1P,Q_1=_\partial Q'\midsp (P\tright Q_1)'' \leq_\partial  \bbox( (\diamondvert P)''\tright (\diamondminus Q_1)''  )''  \rangle$\\
7.\> $\langle P''\leq_1P,Q_1=_\partial Q'\midsp P\tright Q_1 \leq_\partial  \bbox( (\diamondvert P)''\tright (\diamondminus Q_1)''  )''  \rangle$
\end{tabbing}
which is in canonical Sahlqvist form.

Generating the guarded second-order translation and computing the local correspondent are left to the interested reader.
\end{ex}

\begin{thm}
  \label{Sahlqvist thm with implication}
  Every Sahlqvist sequent in the language of DfML has a first-order local correspondent, effectively computable from the input sequent.
\end{thm}
\begin{proof}
Being Sahlqvist, the sequent reduces to canonical Sahlqvist form. But the latter has no occurrences of implication. The only difference is that the binary diamond operators $\odot$, or $\tright$ may occur in the simple Sahlqvist inequality at the end of the reduction process. 
\end{proof}

\section{Restrictions, Extensions and Related Research}
\label{etc}

\subsection{Classical Sahlqvist--Van Benthem Correspondence}
\label{kripke case}
In defining the syntax of sorted modal logic, we certainly assume that the sets of propositional variables $\{P_i\midsp i\in\mathbb{N}\}$ and $\{P^i\midsp i\in\mathbb{N}\}$  are distinct. Undoing sorting, we let them be one and the same set, with $P_i=P^i$ for all $i\in\mathbb{N}$. The distinction between sorts collapses and ceases to exist. The two proof relations $\proves$ and $\vproves$ become essentially identified, since $\psi\vproves\varphi$ iff $\neg\varphi\proves\neg\psi$. $(\;)'$ appears as a negation operator, Galois connected with itself. If the semantics is given in frames $(s,Z,I,\ldots)$ where $s=\{1,\partial\}$, $Z_1=Z_\partial$ and $I\subseteq Z_1\times Z_\partial$ is the identity relation (classical Kripke frames, consult \cite[Remark~3.4, Remark~3.9]{dfnmlA}), then $(\;)'$ is interprted as the set-complement operator. This has the further consequences that every subset is stable, since $U=--U$, and as a simple calculation shows we also have $R''_\Box=R_\Box$, $R''_\Diamond=R_\Diamond$. Furthermore, for any propositional variable $P$ we have a semantic equivalence of $\bb P$ with each of $(\diamondminus P')'$ and $(\diamondvert P')'$, which in the classical Kripke case are $\neg(\diamondminus (\neg P))$ and similarly for $\diamondvert$. Consequently, the two diamond operators $\diamondvert, \diamondminus$ collapse to just one. Semantically, from $\ldvert\{x\}=\ldminus\{x\}$, for any $x$, it follows that $R_\Diamond=R_\Box$. 

For implication, in the classical Kripke frame case the double dual $R^{111}$ coincides with $T^{\partial 1\partial}$ (modulo a permutation of variables that was involved in defining $R^{111}$) and, with a thereby distributive frame logic, $uR^{111}xz$ iff $x\leq u$ and $z\leq u$ holds (consult \cite[Section~7]{dfnmlA} and \cite{choiceFreeHA} for details). The set image operator $\bigodot$ generated by $R^{111}$ is identified with intersection. With the identification of $R^{111}$ and $T^{\partial 1\partial}$  it follows that $\ltright=\bigodot$, hence both are intersection and thus implication is defined as in the classical case, since $A\rspoon C=(A\ltright C')'=-(A\cap(-C))$.

None of the above special assumptions about sorting, frame structure and semantic interpretation affects in the least the way the Sahlqvist - Van Benthem algorithm computes local correspondents.  Guarded second-order translation collapses to classical second-order translation since, by the fact that every set is Galois stable, the $\mathrm{t}${-}invariance constraint is trivially true (see also \cite{vb}, Final Comments section). The only difference is that the resulting correspondent can be now significantly simplified. 

Conclude that if the modal logic is classical and frames are then classical Kripke frames, as explained above, then the generalized Sahlqvist -- Van Benthem correspondence argument we presented collapses to the classical result for Sahlqvist implications. 

\subsection{Correspondence for Substructural Logics}
\label{suzuki case}
Consider  the language of (modal) substructural logics (with weak negation)
\[
\mathcal{L}\ni\varphi\;:=\;p_i(i\in\mathbb{N})\midsp\top\midsp\bot\midsp\varphi\wedge\varphi\midsp \varphi\vee\varphi\midsp\bb\varphi\midsp\dd\varphi\midsp\ttdown\varphi\midsp\varphi\lfspoon\varphi\midsp\varphi\rfspoon\varphi\midsp\varphi\ccirc\varphi\midsp\mathbf{t}.
\]
The Lambek product operator $\ccirc$ is interpreted as $\bigovert$ (consult Section~\ref{frames section}), the closure of the binary image operator $\bigodot$ generated by the double dual $R^{111}$ of the frame relation $T^{\partial 1\partial}$. Consult \cite{redm,choiceFreeHA} for details on modeling reverse implication $\lfspoon$. Frames need to be slightly richer, and we assume a distinguished point $\type{t\!t}\in Z_1$, letting $\mathbf{t}$ be interpreted as the closure $\{\type{t\!t}\}''=\Gamma\type{t\!t}=\mathbb{I}$. An axiom to the effect that $\mathbb{I}$ is a two-sided unit for the product operator $\bigovert$, $\mathbb{I}\bigovert A=A=A\bigovert\mathbb{I}$, needs to be added as well (see \cite{redm} for details).

The reader may notice that not much else changes in the correspondence argument we have provided. For example, the Exchange Rule amounts to commutativity of $\ccirc$, i.e. the sequent $p\ccirc q\proves q\ccirc p$ is to be added. Translation and reduction produce the system $\lset P''\leq P, Q''\leq Q\midsp (P\odot Q)''\leq_1(Q\odot P)''\rset$, which reduces, by (R3), to $\lset P''\leq P, Q''\leq Q\midsp P\odot Q\leq_1(Q\odot P)''\rset$, which is in canonical Sahlqvist form. A first-order correspondent can be now calculated.

The language may be extended with the Grishin operators $\leftharpoondown,\ast,\rightharpoondown$, but we write $\boxast$ for $\ast$, to emphasize that it is really a binary box operator. In Lambek-Grishin algebras $\boxast$ is co-resoduated with $\leftharpoondown$ and $\rightharpoondown$, i.e. the condition $a\geq c\leftharpoondown b$ iff $a\boxast b\geq c$ iff $b\geq a\rightharpoondown c$.  The (representation and) duality argument of \cite{duality2} yields sorted frames in which both the full complex algebra $\gpsi$ and its dual $\gphi$ are residuated. Given the residuation structure  $B\subseteq D\mapsto G$ iff $D{\Large \oast}B\subseteq G$ iff $B\subseteq G\leftmapsto B$ in $\gphi$, the Galois dual structure obtained by defining $A\boxast C=(A'\oast C')'$, $A\Rrightarrow C=(A'\mapsto C')'$ and $C\Lleftarrow A=(C'\leftmapsto A')'$ furnishes $\gpsi$ with a co-residuation (Grishin) structure $A\supseteq F\Lleftarrow C$ iff $A\boxast C\supseteq F$ iff $C\supseteq A\Rrightarrow F$. Issues to be clarified by the interested reader exist, but the principle of extending the correspondence argument is clear.

\begin{rem}
Sahlqvist correspondence and canonicity for substructural logics, in the language
\[
\varphi:= p\midsp\mathbf{t}\midsp\mathbf{f}\midsp\varphi\wedge\varphi\midsp\varphi\vee\varphi\midsp\varphi\leftarrow\varphi\midsp\varphi\circ\varphi\midsp \varphi\rightarrow\varphi,
\] 
have been studied by Suzuki \cite{Suzuki_2011,SUZUKI_2013}, building on some of the insights of Ghilardi and Meloni's constructive canonicity for non-classical logics \cite{GHILARDI1997}. Suzuki identifies subclasses of sentences in the above syntax, designated by $\varphi_\cup,\varphi_\cap,\varphi_\vee,\varphi_\wedge$, defined by  mutual recursion. Sequents $\varphi_\cup\proves\psi_\cap$ are said to have  consistent variable occurrence, while other combinations for the antecedent and consequent of a sequent reveal that a critical subformula must exist. A detailed syntactic and semantic analysis is given in \cite{SUZUKI_2013} (rather extensive to be briefly reviewed), and a correspondence algorithm is proposed. 

In \cite{Suzuki_2011} Suzuki presented results on canonicity for substructural logics and the syntactic approach of \cite{Suzuki_2011} was carried over and adapted for the study of correspondence in \cite{SUZUKI_2013}.
\end{rem}
\subsection{Quantifier Elimination by Ackermann's Lemma}
\label{conradie case}
Second-Order quantifier elimination and its applications in logic, in particular in correspondence theory, is well studied \cite{quantifier-elimination} and two basic algorithms, SCAN and DLS were first  proposed for the task. In a series of articles starting with \cite{Conradie-Goranko-Vakarelov-I}, Conradie, Goranko and Vakarelov introduced the algorithm SQEMA for computing first-order correspondents and which avoids some of the technical difficulties of SCAN and DLS relating to Skolemization. In \cite{Conradie-Goranko-Vakarelov-I} the authors treat the correspondence and canonicity problem for classical modal logic, presenting the core algorithm which uses a modal version of Ackermann's lemma, with the algorithm and the application area variously extended in the sequel in \cite{Conradie-Goranko-Vakarelov-II,Conradie-Goranko-Vakarelov-III, Conradie-Goranko-IV,Conradie-Goranko-Vakarelov-V}, though distribution was assumed in the logics where the approach was applied.

As mentioned already in the Introduction, 
an algebraic and order-theoretic, unified approach to correspondence and canonicity theory was founded in \cite{Conradie-unified}, by Conradie, Ghilardi and Palmigiano. Three of the tools that can be singled out in this approach are, first the use of Ackermann's Lemma, second, a systematic extension of the target language by the addition of adjoint operators (residuals, or Galois connected operators) and third a use of a hybrid language with nominals and co-nominals, used to eliminate second-order variables. 

Directly relating to the content of the present article is \cite{conradie-palmigiano}, where the correspondence problem for non-distributive logics was presented, which can be then specialized to distribution-free modal logics, as well. The extended language according to \cite{conradie-palmigiano}, if DfML is the intended  application, would then be
\begin{eqnarray*}
\widehat{\mathcal{L}}\ni\varphi\;:=& \;\;p_i(i\in\mathbb{N})\midsp{\mathbf{j}\;(\mathbf{j}\in\mathrm{NOM}) \midsp\mathbf{m}\;(\mathbf{m}\in\mathrm{CO{-}NOM})} \midsp\top\midsp\bot\midsp
 \midsp\varphi\wedge\varphi\midsp \\ & \midsp \varphi\vee\varphi\midsp\bb\varphi\midsp\blackdiamond\varphi\midsp\dd\varphi\midsp\bbox\varphi\midsp\ttdown\varphi \midsp\blacktriangledown\varphi \midsp\varphi\rfspoon\varphi\midsp\varphi\bullet\varphi\midsp\varphi\lfspoon\varphi,
\end{eqnarray*}
where the  meaning of the added operators is fixed by introducing the relevant additional adjunction rules to the proof system, such as 
\small{
$\infrule{\varphi\proves\ttdown\psi}{\overline{\psi\proves\blacktriangledown\varphi}},
\hskip2mm
\infrule{\varphi\proves\psi\rfspoon\vartheta}{\overline{\psi\proves\vartheta\lfspoon\varphi}}$
and
$\infrule{\varphi\proves\psi\rfspoon\vartheta}{\overline{\;\;\psi\bullet\varphi\proves\vartheta\;\;}}$.
}
\normalsize

The intended interpretation in \cite{conradie-palmigiano} considers perfect LEs $\mathbb{A}$ from which relational structures $(J^\infty(\mathbb{A}),\leq,M^\infty(\mathbb{A}),\bb^\pi,\dd^\sigma,\ttdown^\pi,\rfspoon^\pi)$ can be extracted, providing relational semantics for the logic. 
The scope of the algebraic approach to correspondence is large, both correspondence and canonicity are subject to this approach and a number of authors contribute to the related research. For details, we direct the reader to \cite{Conradie-unified} and to the subsequent publications by the authors.

A point that is worth making relates to the axioms used by Dunn to axiomatize positive modal logic (PML), (K1) $\Diamond p\wedge\Box q\proves\Diamond(p\wedge q)$ and (K2) $\Box(p\vee q)\proves\Diamond p\vee\Box q$. As mentioned in Example~\ref{K1-K2}, in the algebraic correspondence approach of \cite{conradie-palmigiano} both axioms are Sahlqvist if distribution is assumed and both fail to be Sahlqvist  in distribution-free systems. It is interesting that, for distribution-free modal logic, both (K1) and (K2) can be reduced to canonical Sahlqvist form (consult Example~\ref{K1-K2}), from which first-order correspondents can be calculated, hence some useful complementarity of the approaches appears to be in place.

\bibliographystyle{plain}

\end{document}